\documentclass[conference]{IEEEtran}
\IEEEoverridecommandlockouts
\usepackage{mathtools}
\usepackage{epsfig,endnotes}
\usepackage{amsmath,amssymb,amsfonts}
\usepackage{amsthm}
\usepackage[ruled, vlined,linesnumbered]{algorithm2e}
\usepackage{
  color,
  float,
  graphicx
}
\usepackage{textcomp}
\usepackage{xcolor}
\usepackage{booktabs}
\usepackage{enumitem}
\usepackage{verbatim}
\usepackage{subcaption}
\usepackage{hyperref}
\usepackage{balance}
\usepackage{ifthen}
\usepackage{multicol,xparse,environ}
\usepackage{bm}
\usepackage{dsfont}
\usepackage{multirow}
\usepackage{adjustbox}
\usepackage{pgfplots}\pgfplotsset{compat=1.9}
\usepackage{pgfkeys}
\usepackage{scalefnt}
\newcommand\mtlarge{\fontsize{15pt}{15pt}\selectfont}
\tikzset{
  font=\mtlarge, line width=1.0pt, mark size=3.0pt}
\usepgfplotslibrary{groupplots,dateplot}
\usetikzlibrary{patterns,shapes.arrows}
\pgfplotsset{compat=1.9}
\definecolor{blue1}{RGB}{0,0,255}
\definecolor{darkorange1}{RGB}{255,127,0}
\definecolor{darkgreen}{RGB}{0,100,0}
\definecolor{darkorchid}{RGB}{153,50,204}
\definecolor{cadetblue1}{RGB}{152,245,255}
\definecolor{goldenrod1}{RGB}{255,193,37}
\definecolor{deeppink2}{RGB}{238,18,137}
\definecolor{hotpink255129192}{RGB}{255,129,192}
\definecolor{paleturquoise}{RGB}{175,238,238}
\definecolor{springgreen}{RGB}{0,255,127}
\definecolor{mediumspringgreen}{RGB}{0,250,154}
\definecolor{brown3}{RGB}{205,51,51}
\definecolor{brick}{RGB}{156,102,31}
\definecolor{red1}{RGB}{255,0,0}
\definecolor{cadetblue3}{RGB}{122,197,205}
\definecolor{lime02380}{RGB}{0,238,0}
\newenvironment{customlegend}[1][]{%
    \begingroup
    \csname pgfplots@init@cleared@structures\endcsname
    \pgfplotsset{#1}%
}{%
\csname pgfplots@createlegend\endcsname
\endgroup
}%
\def\addlegendimage{\csname pgfplots@addlegendimage\endcsname}

\newtheorem{definition}{Definition}
\newtheorem{lemma}{Lemma}
\newtheorem{theorem}{Theorem}

\newtheorem{fact}{Fact}
\DeclareMathOperator*{\argmin}{argmin}

\def\BibTeX{{\rm B\kern-.05em{\sc i\kern-.025em b}\kern-.08em
    T\kern-.1667em\lower.7ex\hbox{E}\kern-.125emX}}

\begin{document}

\title{GCON: Differentially Private Graph Convolutional Network via Objective Perturbation
}

\author{\IEEEauthorblockN{Jianxin Wei\textsuperscript{*},
Yizheng Zhu\textsuperscript{*},
Xiaokui Xiao\textsuperscript{*},
Ergute Bao\textsuperscript{*},
Yin Yang\textsuperscript{\dag},
Kuntai Cai\textsuperscript{*},
Beng Chin Ooi\textsuperscript{*},
}
\IEEEauthorblockA{\textsuperscript{*}National University of Singapore, 
\textsuperscript{\dag}Hamad Bin Khalifa University}
jianxinwei@comp.nus.edu.sg,
zhuyizheng@gmail.com,
xkxiao@nus.edu.sg,
ergute@comp.nus.edu.sg,\\
yyang@hbku.edu.qa,
caikt@comp.nus.edu.sg,
ooibc@comp.nus.edu.sg
}

\maketitle

\begin{abstract}
Graph Convolutional Networks (GCNs) are a popular machine learning model with a wide range of applications in graph analytics, including healthcare, transportation, and finance. However, a GCN trained without privacy protection measures may memorize private interpersonal relationships in the training data through its model parameters. This poses a substantial risk of compromising privacy through link attacks, potentially leading to violations of privacy regulations such as GDPR. To defend against such attacks, a promising approach is to train the GCN with differential privacy (DP), a rigorous framework that provides strong privacy protection by injecting random noise into the training process. However, training a GCN under DP is a highly challenging task. Existing solutions either perturb the graph topology or inject randomness into the graph convolution operations, or overestimate the amount of noise required, resulting in severe distortions of the network's message aggregation and, thus, poor model utility.

Motivated by this, we propose {\sf GCON}, a novel and effective solution for training GCNs with edge differential privacy. {\sf GCON} leverages the classic idea of perturbing the objective function to satisfy DP and maintains an unaltered graph convolution process. Our rigorous theoretical analysis offers tight, closed-form bounds on the sensitivity of the graph convolution results and quantifies the impact of an edge modification on the trained model parameters. Extensive experiments using multiple benchmark datasets across diverse settings demonstrate the consistent superiority of {\sf GCON} over existing solutions.

\end{abstract}


\section{Introduction}\label{sec:Introduction}
Graph analytics has promising applications in various domains such as medicine~\cite{ni2017cross, shi2024dmrnet}, social networks~\cite{akcora2012privacy, kou2020efficient}, finance~\cite{karpman2018learning}, and transportation~\cite{bakalov2015time}. In particular, a graph $G = \left< V, E\right>$ comprises a node set $V$ (often representing entities such as user profiles) and an edge set $E$ (e.g., connections between users). Each node in $V$ is also commonly associated with features and labels (e.g., attributes of users). In many applications, connections between individual nodes can be sensitive information. For instance, on a social network, a user may not want to disclose her/his connection with a particular person, political party, or activist group. In such scenarios, the graph analytics results must not leak the individuals' private connections. 

This paper focuses on the training and release of graph convolutional networks (GCNs)~\cite{kipf2016semi}, which is a popular deep learning solution that achieves state-of-the-art performance for common graph analytic tasks such as node classification~\cite{li2024adafgl}. Similar to other neural network models (NNs), GCNs tend to memorize parts of the training data including sensitive connections in its model parameters, which could be extracted with edge inference attacks, e.g., \cite{wu2022DPGCN, he2021stealing,zhang2022inference}, violating privacy regulations such as GDPR~\cite{hoofnagle2019european} and CCPA~\cite{goldman2020introduction}. 

Defending against such edge inference attacks is highly non-trivial, since the GCN training process is rather complex, and it is difficult to control which parts of the training data are memorized by the model. To address this issue, a rigorous solution is to train the GCN with edge-level privacy protection. As a classic privacy protection standard, differential privacy (DP)~\cite{dwork2006differential}, limits the adversary's confidence in inferring the presence or absence of any individual record in the input dataset. Due to its mathematical rigor, DP has been widely accepted by both academia~\cite{hong2021collecting, xia2021dpgraph} and industry~\cite{apple, near2018uber}. Edge DP, in particular, limits the adversary's confidence in inferring the presence or absence of an edge by introducing random perturbations in the released GCN model parameters, as elaborated in Section~\ref{sec:preliminary:dp}. Such random perturbations adversely impact model performance; hence, the main objective in this paper is to preserve the \textit{utility} of the released GCN model, while satisfying the \textit{edge DP} requirement.

\vspace{1pt}
\noindent\textbf{Challenges.} In the literature, a typical mechanism for enforcing DP involves two key steps: (i) determining the \textit{sensitivity} of the analysis result to be released or an intermediate computation result, and (ii) injecting calibrated random noise according to the sensitivity and the privacy parameters of DP, as described in Section~\ref{sec:preliminary:dp}. In the context of GCN training with edge DP, the sensitivity is the maximum impact of adding or removing an edge in the input graph on the final model parameters or the intermediate results (e.g., the adjacency matrix of the graph, gradients of the model weights, etc.), which, in general, is rather high. We explain the reasoning as follows. During the training process for a GCN, it iteratively conducts graph convolution, (also known as message propagation/aggregation). As we will elaborate in Section~\ref{sec:GCN preliminary}, The operation of graph convolution aggregates features from each node's local neighborhood via its edges, and updates the node’s features with the aggregation results. Specifically, in the first convolution, the presence or absence of an edge influences at least the aggregations of its two endpoints. In the $m$-th convolution ($m>1$), this edge can have far-reaching effects on all ($m-1$)-hop neighbors of its endpoints, affecting the aggregations of $O(k^{m-1})$ nodes, where $k$ is the maximum degree of nodes. Consequently, training an effective GCN under edge DP is a highly challenging task, since the high sensitivity leads to heavy perturbations and, thus, low model utility.

\vspace{1pt}
\noindent\textbf{Limitations of existing solutions.} Although there exist DP algorithms~\cite{abadi2016dpsgd,chaudhuri2011DPERM} for training other types of NNs, such as convolutional neural networks~\cite{wei2022dpis} and complex architectures like ResNet~\cite{de2022unlocking} and BERT~\cite{rust2023differential}, these solutions are mainly designed for the case where data records have no interconnections, e.g., images or textual documents, where adding or removing one record only influences the computations on the record itself (e.g., the individual gradient). However, applying these algorithms to training GCNs under edge DP would result in overwhelming noise and poor model utility. We explain this through an example of a classic algorithm DP-SGD~\cite{abadi2016dpsgd}, which consists of two steps. First, for each record, DP-SGD computes the gradient of the loss function to the model parameters, and artificially clips each gradient's $\mathcal{L}_2$-norm by $\tau$, a hyperparameter. Second, it injects noise proportional to $\tau$ into the sum of the gradients and updates the model parameters by the noisy gradient sum. Accordingly, in DP-SGD, adding or removing one record is equivalent to adding or removing one gradient, whose sensitivity is bounded by $\tau$. However, when applied to even the simplest 1-layer GCN, 
the sensitivity becomes at least $2\tau$, because one edge influences at least its two endpoints' aggregation results and, thus, the values of two gradients. In general, when applied to an $m$-layer GCN, the sensitivity is amplified by a factor of at least $2k^{m-1}$. Consequently, a deep GCN with a large $m$ would require overwhelmingly high noise to preserve DP, leading to poor model utility, as we demonstrate in Section~\ref{sec:Experiments}.

Currently, there is only a limited collection of works on integrating DP with GCNs. Among these, {\sf DPGCN}~\cite{wu2022DPGCN} injects random noise directly into the adjacency matrix of the input graph to satisfy edge DP. As a result, every element in the adjacency matrix is perturbed, which severely disrupts the message aggregation of nodes in the GCN, leading to poor model utility. A subsequent method {\sf LPGNet}~\cite{kolluri2022lpgnet} attempts to mitigate this problem by iteratively refining and perturbing a smaller matrix that compresses the information of the original adjacency matrix. In another line of research, {\sf GAP}~\cite{sajadmanesh2023gap} and {\sf ProGAP}~\cite{sajadmanesh2023progap} introduce noise into the aggregate features after each round of message aggregation directly, instead of perturbing the adjacency matrix. As shown in our evaluation results, these methods achieve somewhat improved model utility compared to DPGCN; yet, their performance is still far lower compared to the non-private case, especially with more stringent privacy budget parameters. We conjecture that such poor model performance is due to the large noises injected into the message aggregation process, which is the core of GCN. \textit{Can we preserve edge-level differential privacy while maintaining the integrity of the message aggregation process of GCN, to obtain high-utility models?} 

\vspace{1pt} 
\noindent\textbf{Our contributions.} We provide a positive answer to the above question with {\underline{g}raph \underline{c}onvolutional networks via \underline{o}bjective perturbatio\underline{n} ({\sf GCON}), a novel algorithm for training GCNs with edge DP. The main contributions of our work are summarized as follows.

Firstly, we design a complete training framework adapted to objective perturbation that maintains an unaltered message aggregation process in GCNs, unlike previous solutions that perturb this process and degrade model utility. The main idea is to view GCN training as an optimization problem, and inject calibrated random noise into the corresponding objective function directly, extending the classic DPERM framework~\cite{chaudhuri2011DPERM} to the graph domain. Applying objective perturbation to GCNs requires the objective function to exhibit convexity with respect to the network parameters, which is challenging for traditional multi-layer GCNs with nonlinear activation functions between layers, as their loss functions are inherently non-convex. We address this through a simplified model structure, detailed in Section~\ref{sec:SGC}, and a convex, higher-order differentiable objective function with bounded derivatives. Previous research~\cite{wu2019simplifying, he2020lightgcn, qin2023seign, mao2021ultragcn} has shown such simplified GCNs can still achieve high model utility; meanwhile, as shown in the experiments, the reduction of noise required to satisfy edge-DP more than compensates for the performance loss of the simplified structure compared to the original GCNs.

Secondly, we are the first to derive tight, closed-form bounds on the sensitivity of aggregate features under arbitrary hops. 
Similar to DP-SGD discussed earlier, the original DPERM, when applied directly to GCN training, has the problem that it severely overestimates the GCN's 
sensitivity. Specifically, consider a node $i$ with $k$ neighbors. In a 1-layer GCN, an edge $e$ only influences $\frac{1}{k}$ of its aggregate feature. With a deeper GCN, the influence of $e$ on $i$'s more distant neighbors becomes even smaller. Previous methods, including DPERM and DP-SGD, are unable to account for these progressively smaller influences, and they treat those nodes' features or gradients as completely changed, severely overestimating the sensitivity. In contrast, the proposed solution {\sf GCON} analyzes the subtle influences and tightly bounds the sensitivity of the aggregate features by $O(m)\ll O(k^{m-1})$. To control the amount of information aggregated to each node and further reduce the sensitivity, we adopt two advanced aggregation schemes: PPR and APPR, rooted in personalized PageRank~\cite{luo2019efficient}, which are detailed in Section~\ref{sec:SGC}.

Thirdly, the original DPERM is only feasible in scenarios where two neighboring datasets differ by a single record. To handle graph data, we made non-trivial modifications of the original derivation in DPERM to make it applicable in scenarios where all records could potentially change. Moreover, we correct the proof process of the core theory, Theorem 9, in DPERM~\cite{chaudhuri2011DPERM}, transforming the original eigenvalue analysis into a singular value analysis, detailed in Section~\ref{sec:Analysis}. A broader significance of our derivation is that our framework can be applied under various aggregation schemes used in non-private GCNs, not just the original GCN or PPR/APPR, as long as we have the sensitivity of the aggregate features.

Lastly, to demonstrate the generalization capability and usability of our proposed model, we compare {\sf GCON} on 4 real-world datasets with a wide range of homophily ratios against 6 competitors, under varying privacy budgets. 
The results demonstrate that {\sf GCON} consistently and significantly outperforms existing solutions in terms of prediction accuracy while maintaining the same level of privacy protection across a broad spectrum of configurations.
In the following, Section~\ref{sec:Preliminaries} provides the necessary background on GCN and differential privacy. Section~\ref{sec:problem} clarifies our problem setting, while Section~\ref{sec:Solution} presents the proposed solution {\sf GCON}. Section~\ref{sec:Analysis} formally establishes the correctness of {\sf GCON} and analyzes its privacy costs. Section~\ref{sec:Experiments} contains an extensive set of experiments. Section~\ref{sec:Related Work} overviews related work. Finally, Section~\ref{sec:Conclusion} concludes the paper with directions for future work.
\section{Preliminaries}\label{sec:Preliminaries}
We use $\|\bm{x}\|_2$ and $\|\bm{x}\|_1$ to denote the $\mathcal{L}_2$-norm and $\mathcal{L}_1$-norm of a vector $\bm{x}$, respectively, and $\|\bm{M}\|_F$ to denote the Frobenius norm of a matrix $\bm{M}$. Sections~\ref{sec:GCN preliminary} and \ref{sec:SGC} present the general GCN concepts and two specific flavors of GCNs called SGC and PPNP, respectively. Section~\ref{sec:preliminary:dp} presents the formal definition of differential privacy. The definitions of convexity and strong convexity, which are important concepts in this paper, are presented in Appendix~\ref{app:convexity}.

\subsection{Graph Convolutional Networks}\label{sec:GCN preliminary}
Graph Convolutional Networks (GCNs)~\cite{kipf2016semi} have emerged as a prominent neural network architecture for graph analytics. In what follows, we review the setup and basics of GCNs.

Consider an input graph with nodes $V=\{1,2,\ldots,n\}$ and edges $E$. The edge set $E$ is typically represented by an adjacency matrix $\bm{A} \in \{0,1\}^{n \times n}$, where $\bm{A}_{ij}=1$ signifies an edge from node $i$ to node $j$. Let $\bm{D}$ be a diagonal matrix where each diagonal element, $\bm{D}_{ii}$, corresponds to the degree of node $i$. In GCNs, we often transform $\bm{A}$ into a message passing matrix by multiplying the inverse of $\bm{D}$~\cite{chen2020scalable}, written as $\Tilde{\bm{A}}=\bm{D}^{r-1}\bm{A}\bm{D}^{-r}$, $r\in [0,1]$.

We present GCNs in the context of node classification. Specifically, each node $v \in V$ is associated with a $d_0$-dimensional feature vector, as well as a $c$-dimensional one-hot label vector, i.e., if $v$ belongs to the $j$-th class ($1 \leq j \leq c$), then the $j$-th entry in the label vector is set to 1, and the remaining entries are set to 0. We denote
$\bm{X} \in \mathbb{R}^{n \times d_0}$ and $\bm{Y} \in \mathbb{R}^{n \times c}$ as the feature matrix and label matrix of the $n$ nodes, respectively. 

An $m$-layer GCN iteratively computes the output of layer $l$ ($l=1,2,\ldots,m$), denoted as $\bm{Y}^{(l)}$, based on the message passing matrix $\Tilde{\bm{A}}$ and network parameter $\{\bm{\Theta}^{(l)}\}_{l=1}^m$. In particular, we define $\bm{Y}^{(0)}=\bm{X}$. At each layer $l$, the GCN aggregates the information of every node's neighbors by multiplying the normalized adjacency matrix $\Tilde{\bm{A}}$ with the preceding layer's output $\bm{Y}^{(l-1)}$, and the corresponding layer's parameters $\bm{\Theta}^{(l)}$. Formally, the GCN iteratively computes 
\begin{equation}\label{eq:GCN layer l}
\bm{Y}^{(l)}=H_l \left(\Tilde{\bm{A}} \bm{Y}^{(l-1)} \bm{\Theta}^{(l)}\right), 
\end{equation}
for $l=1,2,\ldots,m$. 
Here $H_l$ is some pre-defined activation function, e.g., the sigmoid function $H_l(u)=1/(1+\exp(-u))$ or the linear mapping $H_l(u)=u$. The output of the last layer $\bm{Y}^{(m)}$ is regarded as the prediction for the nodes in the input graph, denoted as $\bm{\hat{Y}}$.

\vspace{3pt}
\noindent\textbf{Training objective.} In a typical use case of GCN, a portion of the labels are used as the training set, say, with $n_1 < n$ node labels. The training objective is to minimize the dissimilarity between the predicted labels and true labels of the labeled nodes in the training set by updating the network parameter $\{\bm{\Theta}^{(l)}\}_{l=1}^m$ (e.g., using stochastic gradient descent with backpropagation~\cite{zhang2020sgd}). The dissimilarity is measured by the loss function, as follows:
\begin{equation}\label{eq:L_Lambda}
\begin{aligned}
    &L_{\Lambda}(\{\bm{\Theta}^{(l)}\}_{l=1}^m;\bm{\hat{Y}},\bm{Y}) \\
    &= \frac{1}{n_1}\sum^{n_1}_i L(\{\bm{\Theta}^{(l)}\}_{l=1}^m;\bm{\hat{y}}_i,\bm{y}_i) +  \frac{\Lambda}{2} \sum_l^m \|\bm{\Theta}^{(l)}\|_F^2, 
\end{aligned}   
\end{equation}
where $\frac{1}{2}\sum_l^m \|\bm{\Theta}^{(l)}\|_F^2$ is a regularization term for the network parameters, $\Lambda$ is a hyperparameter, and the component $L(\{\bm{\Theta}^{(l)}\}_{l=1}^m;\bm{\hat{y}}_i, \bm{y}_i)$ is some function that evaluates the distance between the prediction $\bm{\hat{y}}_i$ and the true label $\bm{y}_i$ for node $i$, where the prediction $\hat{y}$ is computed based on the network parameters $\{\bm{\Theta}^{(l)}\}_{l=1}^m$. 

\subsection{SGC and PPNP}\label{sec:SGC}
Next, we review two popular techniques, namely SGC and PPNP, which were proposed to improve GCN model performance in terms of both efficiency (e.g., memory consumption and computational cost in training and inference) and effectiveness (e.g., measured by prediction accuracy). 

\vspace{6pt}
\noindent \textbf{SGC.} Simple Graph Convolution (SGC)~\cite{wu2019simplifying} simplifies a multi-layer GCN by using a linear mapping, e.g., $H_l(u) = u$, as the activation function in each graph convolution layer defined in Eq. (\ref{eq:GCN layer l}).
This leads to a streamlined GCN which still maintains an identical message propagation pathway as the original GCN. Formally, the model equation of an $m$-layer GCN is simplified to
\begin{equation}\label{eq:SGC}
    \hat{\bm{Y}} = \bm{Y}^{(m)} = H(\Tilde{\bm{A}}^{m} \bm{X} \bm{\Theta}^{(1)}\bm{\Theta}^{(2)}\cdots \bm{\Theta^}{(m)}) = H(\Tilde{\bm{A}}^{m} \bm{X} \bm{\Theta}),
\end{equation}
where we have transformed the matrix product of $\bm{\Theta}^{(1)}\bm{\Theta}^{(2)}\cdots \bm{\Theta^}{(m)}$ to a single matrix $\bm{\Theta}$ of dimension $d_0\times c$. It has been demonstrated in \cite{wu2019simplifying} that SGC improves training efficiency while retaining the accuracy performance of the original GCN.

\vspace{3pt}\noindent
\textbf{PPNP.} Recall from Eq (\ref{eq:GCN layer l}) that in each layer $l$, the traditional aggregation in GCN directly multiplies the previous output $\bm{Y}^{(l-1)}$ with the normalized adjacency matrix $\Tilde{\bm{A}}$, which can be seen as the transition probability matrix in random walks~\cite{ppnp}. PPNP~\cite{ppnp} designs an advanced propagation scheme for improving the performance of GCNs, based on personalized PageRank (PPR)~\cite{luo2019efficient}, instead of iteratively applying $\Tilde{\bm{A}}$ to aggregate features. In particular, PPR introduces the concept of restarting probability to a random walk. Let $\bm{R}_0$ be an identity matrix $\bm{I}$ denoting the roots. The PPR propagation matrix with $m$ propagation steps, denoted as $\bm{R}_m$, is defined as
\begin{equation}\label{eq:PPR recursive}
    \bm{R}_m = (1-\alpha)\Tilde{\bm{A}} \bm{R}_{m-1} + \alpha \bm{I},
\end{equation}
where $\alpha \in (0,1]$ is the restart probability. When $m$ approaches infinity, $\bm{R}_m$ converges to $\bm{R}_{\infty}$, defined as 
\begin{equation}\label{eq:PPR matrix}
    \bm{R}_{\infty} = \lim_{m\rightarrow \infty} \bm{R}_m = \alpha\left(\bm{I}-(1-\alpha) \Tilde{\bm{A}}\right)^{-1} .    
\end{equation}
PPNP initially processes $\bm{X}$ through a feed-forward neural network $M$ denoted as $M(\bm{X})$. Subsequently, it uses $\bm{R}_{\infty}$ to aggregate the processed features $M(\bm{X})$. Formally, its model equation is
\begin{equation*}
    \hat{\bm{Y}} = \bm{R}_{\infty} M(\bm{X}).
\end{equation*}

There remains an efficiency issue with the above approach: many real graphs are sparse, meaning that their message passing matrix $\Tilde{\bm{A}}$ is a sparse matrix; however, the inversion in Eq. (\ref{eq:PPR matrix}) results in a $n\times n$ dense matrix, leading to higher computational and memory costs. To address this issue, the authors of \cite{ppnp} further introduce an approximate version of PPR (APPR) to GCN, which employs $\bm{R}_m$ with finite choices of $m$. In such cases, the corresponding propagation matrix is derived by recursively expanding Eq. (\ref{eq:PPR recursive}).
\begin{equation}\label{eq:APPR matrix}
    \bm{R}_m= \alpha \sum_{i=0}^{m-1} (1 - \alpha)^i  \Tilde{\bm{A}}^i + (1-\alpha)^m \Tilde{\bm{A}}^m.
\end{equation}

\subsection{Differential Privacy}\label{sec:preliminary:dp}
Differential privacy (DP) provides rigorous protection of individuals' private information.
DP ensures that an adversary can only infer with limited confidence the presence or absence of any individual of the input by observing the output. The most commonly used notion of DP is ($\epsilon, \delta$)-DP, defined as follows.

\begin{definition}[($\epsilon, \delta$)-Differential Privacy (DP)~\cite{dwork2006differential}] \label{def:prelim-dp}
A randomized algorithm $\mathcal{A}:\mathcal{D}\rightarrow \mathcal{R}$ satisfies {\it ($\epsilon, \delta$)-DP} if given
any neighboring datasets $D, D'\subseteq \mathcal{D}$, we have
\begin{equation} \label{eq:prelim-DP}
\Pr[\mathcal{A}(D) \in \mathcal{O}] \le \exp({\epsilon})  \cdot \Pr[\mathcal{A}(D') \in \mathcal{O}] + \delta,
\end{equation}
for any set of output $\mathcal{O}\subseteq \mathcal{R}$.
\end{definition}

In Eq. (\ref{eq:prelim-DP}), the \textit{privacy budget}, specified by parameters $\epsilon$ and $\delta$, quantifies the indistinguishability between the two distributions of $\mathcal{A}(D) $ and $\mathcal{A}(D')$ for any neighboring datasets $D$ and $D'$ (graphs in our setting). 
Smaller $\epsilon$ and $\delta$ indicate that the two distributions are more similar, which implies lower confidence for an adversary to infer a private record (an edge in our case) of the input, and vice versa. In practice, $\delta$ is usually set to a small value no more than the inverse of the dataset size (a special case is when $\delta=0$, which is referred to as $\epsilon$-DP), and $\epsilon$ is often used to control the strength of the privacy requirement. The above definition involves the concept of neighboring datasets, whose generic definition is as follows.

\begin{definition}[Neighboring datasets] \label{def:prelim-neighboring}
Datasets $D$ and $D'$ are said to be {\it neighboring} if they differ by one record.
\end{definition}

The above generic definition involves two abstract concepts, \textit{dataset} and \textit{record}. In the context of graph analytics, a common notion is \textit{edge DP}~\cite{kolluri2022lpgnet,sajadmanesh2023gap,sajadmanesh2023progap}, in which the entire input graph $D=\left<V, E, \bm{X}, \bm{Y}\right>$ is considered a dataset, and an edge $e \in E$ is considered a record, i.e., two neighboring graphs $D$ and $D'$ differ by exactly one edge. In this paper, we aim to satisfy edge-DP in the training and release of a GCN model.

\section{Problem Setting}\label{sec:problem}

Our goal is to design an effective algorithm for training GCNs with differential privacy. In our presentation, we focus on node classification as described in Section~\ref{sec:GCN preliminary}; other applications of GCN are left as future work. 
In particular, the algorithm takes a graph dataset $D = \left<V, E,\bm{X},\bm{Y} \right>$ with node set $V$ and edge set $E$, as well as features $\bm{X}$ and labels $\bm{Y}$ of the nodes. It outputs GCN parameters $\Theta$.

In terms of privacy, we focus on edge DP as described in Section~\ref{sec:preliminary:dp}, which prevents an adversary from differentiating between any two neighboring graphs that differ by a single edge. We consider the node set $V$ as well as features/labels of nodes in the training set as public information. 

Formally, given privacy budget $(\epsilon, \delta)$, for any graph dataset $D = \left<V, E, \bm{X}, \bm{Y}\right>$ (which consists of node set $V$, edge set $E$, feature matrix $\bm{X}$, and label matrix $\bm{Y}$) and its edge-level neighboring dataset $D' = \left<V, E', \bm{X}, \bm{Y}\right>$ whose edge set $E'$ differs from $E$ by exactly one edge, the output $\bm{\Theta}$ of algorithm $\mathcal{A}$ satisfies
\begin{equation}\label{eq:graph edge DP}
    \Pr[\mathcal{A}(V,\bm{X},\bm{Y},E) \in \mathcal{O}] \leq e^{\epsilon} \Pr[\mathcal{A}(V,\bm{X},\bm{Y},E') \in \mathcal{O}] + \delta,
\end{equation}
for any $\mathcal{O} \subseteq \textit{Range}(\mathcal{A})$.

\section{GCON}\label{sec:Solution}

\subsection{High-Level Idea and Major Challenges}\label{sec:idea}

In a nutshell, the main idea of the proposed {\sf GCON} algorithm is to inject random noise into the objective function of GCN training, following the DPERM framework. The scale of noise is carefully calibrated to satisfy edge DP with given privacy parameters $\epsilon$ and $\delta$. Note that {\sf GCON} does not directly perturb the graph topology or introduce additional noise during the training process besides objective perturbation. This way, {\sf GCON} avoids the severe disruption of the GCN's message propagation/aggregation process commonly encountered in existing solutions, as explained in Section~\ref{sec:Introduction}. Realizing this idea involves overcoming several major challenges:

\vspace{3pt}
\noindent \textbf{Challenge 1: bounding the sensitivity of the aggregation results.} 
As mentioned in Sections~\ref{sec:Introduction} \& \ref{sec:GCN preliminary}, an edge influences the training of model parameters by influencing the aggregation results. We begin with bounding the sensitivity of the aggregation results. It is crucial since the added noise required to satisfy edge DP is calibrated based on it. This is challenging because an edge affects not only the aggregation of its endpoints but also nodes farther away due to multiple convolutions/layers in the GCN. Specifically, according to Eq. (\ref{eq:GCN layer l}), GCNs iteratively multiply the normalized adjacency matrix with the nodes' representations. Then adding or removing one edge affects the aggregations of at least $2k^{m-1}$ nodes after $m$ convolutions/layers, where $k$ is the maximum degree of nodes, making it difficult to bound the specific amount of the impact on these aggregations.

\vspace{3pt}
\noindent \textbf{Challenge 2: quantifying how the above sensitivity affects model parameters.} Even if we have the sensitivity of the aggregations, we need to quantify how this sensitivity and the noise injected into the objective function affect the (distribution of) optimized model parameters. The quantification involves analyzing the Jacobian matrices of the mapping from the parameters to the noise and calibrating the noise distribution. The original derivation of DPERM~\cite{chaudhuri2011DPERM} isolates a single divergent element within the Jacobian matrices, which is the difference between the neighboring datasets. However, this approach is inadequate for training GCNs, where the Jacobian matrices differ in numerous items. We need to numerically bound the difference between these Jacobian matrices.

\vspace{3pt}
\noindent \textbf{Challenge 3: ensuring strong convexity of the GCN training objective.} According to the literature~\cite{chaudhuri2011DPERM}, the objective perturbation scheme relies on two crucial properties of the objective function: strong convexity (Definition~\ref{def:strongly convex}) and bounded derivatives with respect to the model parameters. Both are difficult for traditional GCNs with multiple layers and nonlinear activation functions between the layers, whose loss function is inherently non-convex.

In addition to the above major challenges, we also face the issue that the feature dimensionality $d$ could significantly affect the privacy guarantee. In particular, a larger $d$ corresponds to more model parameters, which may lead to more leakage of sensitive information.

\subsection{Solution Overview}\label{sec:Solution Overview}
The difficulty of Challenges 1 and 2 lies in the theoretical derivation. We provide a comprehensive proof approach different from the original DPERM, along with several useful lemmas that might be of independent interest, in Section~\ref{sec:Analysis}.
For Challenge 1, in addition to providing a rigorous derivation and tight bound, we further reduce the upper bound of the sensitivity by utilizing the personalized PageRank propagation, PPR, and its approximation APPR, described in Section~\ref{sec:SGC}. 
Unlike conventional message aggregation methods that multiply the adjacency matrix with node features, PPR/APPR control the amount of information aggregated to each node by a tunable restart probability $\alpha$, effectively limiting information propagated through the edges and leading to reduced sensitivity. 
For instance, at $\alpha=1$, nodes propagate nothing to their neighbors; conversely, at $\alpha=0$, the scheme reverts to the standard propagation mechanism employed in conventional GCNs. Later, we present the detailed construction of PPR and APPR propagation matrices in Section~\ref{sec:GCON Propagation}, and derive tight, closed-form bounds on their sensitivities in Lemma~\ref{lemma:Psi_Z of PPR and APPR}.

For Challenge 3, we follow the simplifying GCN approach (called SGC and detailed in Section~\ref{sec:SGC}) that streamlines the GCN architecture by eliminating nonlinearities and consolidating parameter matrices between layers, transforming the GCN into a linear model. Then the training of the GCN can be expressed as a convex optimization problem, making objective perturbation feasible.
SGC has been shown to be effective in various non-DP settings~\cite{he2020lightgcn, mao2021ultragcn}. Moreover, we further improve the model structure (as shown in Eq.(\ref{eq:concatenate Z})) to enhance its performance, following recent variants of SGC \cite{chanpuriya2022asgc,chien2020gprgnn}, which are state-of-the-arts methods in non-private settings. Through our novel use of SGC, it enables the objective perturbation approach with reduced perturbations and supports an arbitrary number of propagation steps, while existing DP methods~\cite{wu2022DPGCN,kolluri2022lpgnet,sajadmanesh2023gap,sajadmanesh2023progap} are unable to realize GCNs' potential since they are only feasible with GCNs with no more than 2 message propagation steps (otherwise, the model utility would be rather poor due to excessive noise). Our solution more than compensates for the slight decline in model capability caused by the simplification, as evidenced by our extensive experiments with various tasks and settings in Section~\ref{sec:Experiments}, where GCON often outperforms the competitors by a large margin and generalizes better. 
Additionally, we select 2 higher-order differentiable loss functions and derive the bounds of their derivatives. The properties of these two functions facilitate the application of objective perturbation, and their performance is comparable to that of the most commonly used cross-entropy loss. Our detailed model architecture is presented later in Section~\ref{sec:Network Structure}.

Finally, to address the dimensionality issue, we present an MLP-based feature encoder that does not utilize any edge information, detailed in Section~\ref{sec:encoder module}.

\subsection{Detailed {\sf GCON} Model}\label{sec:Our Model}

\subsubsection{Feature Encoder}\label{sec:encoder module}
First, we propose a feature encoder based on multilayer perceptron (MLP) to reduce feature dimension. Let the feature matrix of the labeled and unlabeled nodes be $\bm{X}_l \in \mathbb{R}^{n_1 \times d_0}$ and $\bm{X}_{ul} \in \mathbb{R}^{ (n-n_1) \times d_0}$, respectively. The framework of our encoder is as follows.
\begin{equation*}
\begin{aligned}
\overline{\bm{X}_l} &= H_{mlp}\left(\operatorname{MLP}\left(\bm{X}_l ; \bm{W}_1 \right)\right),\\
\overline{\bm{Y}_l} &=H\left(\overline{\bm{X}_l} \bm{W}_2\right),\\
\end{aligned} 
\end{equation*}
where $\bm{W}_1$ and $\bm{W}_2$ are model parameters and $H_{mlp}$ and $H$ are activation functions. Note that in our encoder, we exclusively utilize node features and labels, which are deemed non-private in our problem setting defined in Section~\ref{sec:problem}. The encoder operates by initially transforming the original features $\bm{X_l}$ into a new feature space, resulting in $\overline{\bm{X}_l}\in \mathbb{R}^{n_1 \times d_1}$, where $d_1$ is the new dimension. After this transformation, the encoder predicts the labels of these features, represented as $\overline{\bm{Y}_l}$. Let the real label matrix of $\bm{X}_l$ and $\bm{X}_{ul}$ be $\bm{Y}_l$ and $\bm{Y}_{ul}$ ($\bm{Y}_{ul}$ is unknown in training), respectively. The training of the encoder is carried out through the minimization of a classification loss function $L_{mlp}$ as follows.
\begin{equation*}
    \bm{W}_1^{\star}, \bm{W}_2^{\star}=\argmin _{\bm{W}_1,\bm{W}_2} L_{mlp}(\overline{\bm{Y}_l}, \bm{Y}_l).
\end{equation*}
Upon completion of the training, we encode all node features $\bm{X}$ to new features $\overline{\bm{X}} \in \mathbb{R}^{n \times d_1}$ by $\bm{W}_1^{\star}$ (or $\bm{W}_2^{\star}$) as follows.
\begin{equation*}
\begin{aligned}
\overline{\bm{X}} &= H_{mlp}\left(\operatorname{MLP}(\bm{X} ; \bm{W}_1^{\star} )\right).
\end{aligned} 
\end{equation*} 

The complete algorithm of the encoder is outlined in Algorithm~\ref{alg:encoder} in Appendix~\ref{app:Algorithms}. Note that the encoder does not involve any private edge information, and, thus, preserves edge privacy automatically. 

For simplicity, in the following sections, we use $\bm{X} \in \mathbb{R}^{n \times d_1}$ and $\bm{Y} \in \mathbb{R}^{n \times c}$ to denote the feature matrix and the label matrix, respectively, instead of $\overline{\bm{X}}$ and $\overline{\bm{Y}}$.

\subsubsection{Propagation with PPR and APPR}\label{sec:GCON Propagation}
We employ two advanced propagation schemes in {\sf GCON}, PPR and APPR. 
Let $\hat{\bm{A}} = \bm{A} + \bm{I}$ be the adjacency matrix with self-loops. Let $\bm{D}$ be the diagonal degree matrix of $\hat{\bm{A}}$, i.e., $\bm{D}_{ii}=\sum_k^n \hat{\bm{A}}_{i k}$ and $\bm{D}_{ij}=0$ for $i\neq j$. We follow the normalization for $\Tilde{\bm{A}}$ in \cite{chen2020scalable} that sets $\Tilde{\bm{A}} = \mathbf{D}^{r-1} \hat{\bm{A}} \mathbf{D}^{-r}$ where $r\in [0,1]$. In this work we set $r=0$ and obtain $\Tilde{\bm{A}}=\mathbf{D}^{-1} \hat{\bm{A}}$. 
Given $\Tilde{\bm{A}}$, we construct our propagation matrix $\bm{R}_m$, $m\in [0, \infty]$, utilizing the PPR and APPR schemes (Eq. (\ref{eq:PPR matrix}) \& (\ref{eq:APPR matrix})), as follows.
\begin{equation}\label{eq:R_m}
    \bm{R}_m = \begin{cases}
        \bm{I} , & m=0,\\
        \alpha \sum_{i=0}^{m-1} (1 - \alpha)^i  \Tilde{\bm{A}}^i + (1-\alpha)^m \Tilde{\bm{A}}^m , & m\in(0, \infty) \\
        \alpha\left(\bm{I}-(1-\alpha) \Tilde{\bm{A}}\right)^{-1} , & m=\infty
    \end{cases}
\end{equation}
In the above equation, $\bm{R}_\infty$ and $\bm{R}_m$, with $m\in(0, \infty)$, represent the PPR and APPR schemes, respectively. For simplicity, we no longer emphasize which specific scheme {\sf GCON} uses unless necessary. Instead, we represent both schemes with $\bm{R}_m$, where $m$ ranges from 0 to $\infty$.
We prove that the matrix $\bm{I}-(1-\alpha) \Tilde{\bm{A}}$ always has an inverse in Lemma~\ref{lemma:invertible} in Appendix~\ref{app:lemma invertible}.

\subsubsection{SGC-Based Network Architecture}\label{sec:Network Structure}
Before propagation, we normalize the $\mathcal{L}_2$-norm of each row in $\bm{X}$ (i.e., each feature) to be 1. In the step of message propagation, we multiply the propagation matrix $\bm{R}_m$ with the feature matrix $\bm{X}$ and obtain the aggregate feature matrix
\begin{equation}\label{eq:Z_m}
    \bm{Z}_m =\bm{R}_m\bm{X}.
\end{equation}
Furthermore, to enhance the model performance, we augment the features by using several different propagation steps $m_1, m_2, \cdots, m_s \in [0, \infty]$ and concatenate them as 
\begin{equation}\label{eq:concatenate Z}
\bm{Z}=\frac{1}{s}(\bm{Z}_{m_1} \oplus \bm{Z}_{m_2} \oplus \cdots \oplus \bm{Z}_{m_s}),
\end{equation}
where $\oplus$ represents concatenation by row. We weight the concatenated features by $\frac{1}{s}$ to bound the $L_2$-norm of each row of $\bm{Z}$. We then input $\bm{Z}$ into our GCN. 

For the network architecture, we follow SGC idea (outlined in Section~\ref{sec:SGC}) and build a single-layer network with linear mapping. Let $\bm{\Theta} \in \mathbb{R}^{d \times c}$ be the 1-layer network parameters, where $d=sd_1$. Formally, the model makes prediction $\hat{\bm{Y}}$ by computing $\hat{\bm{Y}} = \bm{Z} \bm{\Theta}$.

\subsubsection{Strongly-Convex Loss Function}\label{sec:GCON loss}
The formulation of the original loss function $L_{\Lambda}(\bm{\Theta};\bm{Z},\bm{Y})$ is shown in Eq. (\ref{eq:L_Lambda}). Let $\bm{z}_i$ and $\hat{\bm{y}}_i$ denote the $i$-th columns of $\bm{Z}^T$ and $\hat{\bm{Y}}^T$, respectively. In Eq. (\ref{eq:L_Lambda}), for the convenience of subsequent analysis, we define $L(\bm{\Theta};\hat{\bm{y}}_i,\bm{y}_i)$ as either the MultiLabel Soft Margin Loss~\cite{paszke2019pytorch} or the pseudo-Huber Loss~\cite{2020SciPy-NMeth}, both of which are convex w.r.t. $\bm{\Theta}$. The specific mathematical expressions of the two losses are shown in Appendix~\ref{app:derivatives of l(x;y)}. As in our solution, $\hat{\bm{y}}_i$ is obtained by $\bm{Z}$ and $\bm{\Theta}$, here we denote it as $L(\bm{\Theta};\bm{z}_i,\bm{y}_i)$. Both loss functions evaluate and sum the discrepancies between each element of $\bm{y}_i$ (denoted as $\bm{y}_{ij}, j\in[1,c]$) and its corresponding prediction, $\bm{z}_i^T \bm{\theta}_j$, where $\bm{\theta}_j$ is the $j$-th column of $\bm{\Theta}$. Formally,
\begin{equation}\label{eq:loss function component}
\begin{aligned}
    L(\bm{\Theta};\bm{z}_i,\bm{y}_i) &= \sum_j^{c} \ell(\bm{z}_i^T \bm{\theta}_j; \bm{y}_{ij}) ,\\
\end{aligned}
\end{equation}
Under the above setting, $L_{\Lambda}(\bm{\Theta};\bm{Z},\bm{Y})$ is strongly convex w.r.t. $\bm{\Theta}$ (formally proved in Lemma~\ref{lemma:convexity} in Appendix~\ref{app:convexity}), which is essential for our subsequent privacy analysis. 

\subsubsection{Objective Perturbation}\label{sec:GCON pert}
Next, we perturb the above strongly convex loss function to satisfy edge DP by adding two perturbation terms. Formally, the noisy loss function $L_{priv}(\bm{\Theta};\bm{Z},\bm{Y})$ is
\begin{equation}\label{eq:loss_priv}
\begin{aligned}
    &L_{priv}(\bm{\Theta};\bm{Z},\bm{Y})
    = L_{\Lambda}(\bm{\Theta};\bm{Z},\bm{Y})
    +\frac{1}{n_1} \bm{B} \odot \bm{\Theta} + \frac{1}{2}\Lambda' \|\bm{\Theta}\|_F^2 \\
    =& \frac{1}{n_1}\sum^{n_1}_i L(\bm{\Theta};\bm{z}_i,\bm{y}_i) + \frac{1}{2}\Lambda \|\bm{\Theta}\|_F^2 +\frac{1}{n_1} \bm{B} \odot \bm{\Theta} + \frac{1}{2}\Lambda' \|\bm{\Theta}\|_F^2,
\end{aligned}
\end{equation}
where $\Lambda'$ is a parameter computed by Eq. (\ref{eq:Lambda_prime}) (detailed later in Section~\ref{sec:Analysis}), $\odot$ denotes an element-wise product followed by a summation of the resulting elements, and $\bm{B} = (\bm{b}_1,\bm{b}_2,\cdots, \bm{b}_c)$ is a noise matrix. 
The columns $\bm{b}_1, \bm{b}_2, \cdots, \bm{b}_c$ of $\bm{B}$ are independently and uniformly sampled on a $d$-dimensional sphere with a radius following the Erlang distribution whose probability density function (PDF) is
\begin{equation}\label{eq:norm b distribution}
    \gamma(x) = \frac{x^{d-1} e^{-\beta x} \beta^{d}}{(d-1)!}.
\end{equation}
The detailed sampling algorithm is shown in Algorithm~\ref{alg:sampling}.
In the above equation, $\beta$ is a parameter computed by Eq. (\ref{eq:beta}) (detailed later in Section~\ref{sec:Analysis}).

We minimize $L_{priv}(\bm{\Theta};\bm{Z},\bm{Y})$ to obtain the network $\bm{\Theta}_{priv}$, which is released to the analyst/adversary. Formally, 
\begin{equation}\label{eq:Theta_priv}
    \bm{\Theta}_{priv} = \argmin_{\bm{\Theta}} L_{priv}(\bm{\Theta};\bm{Z},\bm{Y}).
\end{equation}

\subsubsection{Inference}\label{sec:GCON inference}
We then utilize the obtained network $\bm{\Theta}{priv}$ to perform inference for the target unlabeled nodes. We consider the setup where the server publishes the trained model and then an untrusted user, represented as an unlabeled node within the target nodes, queries the model and receives the predicted label. The target nodes might belong to either (i) the same graph as the training data or (ii) a different graph with possibly DP requirement. 

We focus on scenario (i) in our evaluation as it is the standard evaluation setup in the literature~\cite{kolluri2022lpgnet,sajadmanesh2023gap,sajadmanesh2023progap}. As the untrusted user knows her/his edge information, we do not add noise when utilizing the information of her/his direct neighbors, while edges other than its own are private. 
Accordingly, we devise an inference approach that aggregates the features of directly connected nodes and generates predictions $\hat{\bm{Y}}$, as shown in Eq. (\ref{eq:private R}), where $\alpha_I$ is a hyperparameter ranging in $[0,1]$ (the restart probability at the inference stage). 
Note that this approach involves only the edges directly connected to the query node, which are assumed to be known to the untrusted user, and no other edges in the graph. Hence, this approach satisfies edge DP as it does not reveal any additional private edge information of non-neighboring nodes to the querying node.
\begin{equation}\label{eq:private R}
\begin{aligned}
\hat{\bm{R}}_{m_i} &=
    \begin{cases}
     \bm{I}, & m_i = 0\\
    (1-\alpha_I)\Tilde{\bm{A}} + \alpha_I \bm{I}, & m_i > 0
    \end{cases} i\in[1,s] \\
    \hat{\bm{Y}} &= (\hat{\bm{R}}_{m_1}\bm{X} \oplus \hat{\bm{R}}_{m_2}\bm{X} \oplus \cdots \oplus \hat{\bm{R}}_{m_s}\bm{X}) \bm{\Theta}_{priv} 
 \end{aligned}     
\end{equation}
In scenario (ii), if the edges in the different graphs are private, we employ the same inference approach as in Eq. (\ref{eq:private R}). If the edges are public, we can directly compute $\bm{Z}$ without a DP mechanism, and derive labels by $\hat{\bm{Y}} = \bm{Z} \bm{\Theta}_{priv}$. We present the corresponding procedure as in Algorithm~\ref{alg:Inference}.

\subsection{The Complete Algorithm}\label{sec:training algorithm}
The complete training algorithm of {\sf GCON} is shown in Algorithm~\ref{alg:Train}. In the input, $\epsilon$ and $\delta$ represent the privacy constraints; $d_1$, $\alpha$, $\{m_i\}_{i=1}^s$, and $\Lambda$ are tunable hyperparameters. In Line~\ref{line:encode xy}, we preprocess the feature matrix $\bm{X}$ and label matrix $\bm{Y}$ by the feature encoder presented in Section~\ref{sec:encoder module}. In Line~\ref{line:normalize X}, the $\mathcal{L}_2$-norm of each feature is normalized to 1.
In Lines~\ref{line:propagation matrices}-\ref{line:aggregate matrices}, we construct the propagation matrices $\{\bm{R}_{m_i}\}_{i=1}^s$ and $\{\bm{Z}_{m_i}\}_{i=1}^s$. In Line~\ref{line:concatenate Z}, we construct the augmented feature matrix $\bm{Z}$.
In Line~\ref{line:compute parameters}, we compute $\Lambda'$ and $\beta$, which sets up the stage for the objective perturbation (the detailed computations are explained in Section~\ref{sec:Analysis}).
In Line~\ref{line:sample B}, we sample the perturbation noise by running Algorithm~\ref{alg:sampling}.
In Line~\ref{line:loss_priv}, we construct the noisy objective function, based on the sampled noise. In Line~\ref{line:optimize loss_priv}, we minimize the objective function and obtain the optimal network parameters. 

\begin{algorithm}[ht]\small
\caption{The training algorithm for {\sf GCON}}
\label{alg:Train}
\KwIn{Training dataset $D=\left< V,E,\bm{X},\bm{Y} \right>$, new feature dimension $d_1$, loss function $\ell$, restart probability $\alpha$, propagation steps $m_1, m_2, \cdots, m_s$, regularization coefficient $\Lambda$, privacy parameters $\epsilon, \delta$, and budget allocator $\omega$. }
\KwOut{Approximate minimizer $\bm{\Theta}_{priv}$.}
$\bm{X} = \text{FeatureEncoder}(\bm{X}, \bm{Y}, d_1)$\;\label{line:encode xy}
Normalize the $\mathcal{L}_2$-norm of each row in $\bm{X}$ to 1\;\label{line:normalize X}
Initialize model parameters $\bm{\Theta}$\;
$\Tilde{\bm{A}}=\mathbf{D}^{-1} (\bm{A} + \bm{I})$ \;\label{line:normalized adj}
Compute $\bm{R}_{m_1}, \bm{R}_{m_2}, \cdots, \bm{R}_{m_s}$ by Eq. (\ref{eq:R_m})\; \label{line:propagation matrices}
Compute $\bm{Z}_{m_1}, \bm{Z}_{m_2}, \cdots, \bm{Z}_{m_s}$ by Eq. (\ref{eq:Z_m})\; \label{line:aggregate matrices}
$\bm{Z}=\frac{1}{s}(\bm{Z}_{m_1}, \bm{Z}_{m_2}, \cdots, \bm{Z}_{m_s})$\;\label{line:concatenate Z}
Compute $\Lambda'$, and $\beta$ by Eq. (\ref{eq:Lambda_prime}) and (\ref{eq:beta}) \;\label{line:compute parameters}
Construct $\bm{B}$, whose columns $\bm{b}_1,\bm{b}_2,\cdots, \bm{b}_c$ are independently drawing by Algorithm~\ref{alg:sampling} with inputs $d,\beta$ \;\label{line:sample B}
$L_{priv}(\bm{\Theta};\bm{Z},\bm{Y})= \frac{1}{n_1}\sum^{n_1}_i L(\bm{\Theta};\bm{z}_i,\bm{y}_i) + \frac{1}{2}\Lambda \|\bm{\Theta}\|_F^2 +\frac{1}{n_1} \bm{B} \odot \bm{\Theta} + \frac{1}{2}\Lambda' \|\bm{\Theta}\|_F^2$;
\label{line:loss_priv}

$\bm{\Theta}_{priv} = \argmin L_{priv}$\;\label{line:optimize loss_priv}
\Return $\bm{\Theta}_{priv}$\;
\end{algorithm}
\section{Analysis}\label{sec:Analysis} 
This section presents the privacy analysis of Algorithm~\ref{alg:Train}. We formalize the guarantees as the following theorem, whose detailed proof is in Appendix~\ref{app:proof of Theorem}.
\begin{theorem}\label{theorem:train DP}
Algorithm~\ref{alg:Train} satisfies ($\epsilon,\delta$)-DP, with $\Lambda'$ and $\beta$ set to be Eq.~(\ref{eq:Lambda_prime}) - (\ref{eq:beta}), 
\begin{align}
    &\Lambda' = \begin{cases}
        0, & \textrm{if } \epsilon_{\Lambda} \leq (1-\omega)\epsilon \\
        \frac{c(2c_2 + c_3 c_{\theta}) \Psi(\bm{Z})}{n_1(1-\omega)\epsilon} - \Lambda, & \textrm{otherwise}. 
    \end{cases}\label{eq:Lambda_prime}\\
    &\beta = \frac{ \max(\epsilon - \epsilon_{\Lambda}, \omega \epsilon) }{c (c_1 + c_2 c_{\theta}) \Psi(\bm{Z})} .\label{eq:beta}
\end{align}
where $c, d, \Lambda, \ell(x;y), n_1, \omega, \alpha, \{m_i\}_{i=1}^s$ are inputs and other parameters $c_1,c_2,c_3, \Psi(\bm{Z}), c_{sf}, c_{\theta}, \epsilon_{\Lambda}$ are sequentially computed by Eq.~(\ref{eq:supremum of l(x;y)}) - (\ref{eq:epsilon_Lambda}).
\begin{align}
    &c_1 = \sup(|\ell'|), c_2 = \sup(|\ell''|), c_3 = \sup(|\ell'''|). \label{eq:supremum of l(x;y)}\\
    &\Psi(\bm{Z}) = \frac{1}{s} \sum_i^s \frac{2(1-\alpha)}{\alpha}[1-(1-\alpha)^{m_i}].\label{eq:Psi_Z in Theorem}\\
    &c_{sf} = \min\{u \textrm{ s.t. } u>0 \textrm{ and}  \int_0^u \frac{x^{d-1} e^{-x}}{(d-1)!} dx \geq 1 - \frac{\delta}{c}\}. \label{eq:find minimal c_sf} \\
    &\Lambda = \max\left(\Lambda, \frac{c c_2 \Psi(\bm{Z}) c_{sf}}{n_1 \omega\epsilon}+\xi\right),  \xi \in \mathbb{R}^{+} .\label{eq:Lambda constraint}\\
    &c_{\theta} = \frac {n_1 \omega\epsilon c_1 + c c_1 \Psi(\bm{Z}) c_{sf}} {n_1 \omega\epsilon \Lambda - c  c_2 \Psi(\bm{Z}) c_{sf}} .\label{eq:c_theta}\\
    &\epsilon_{\Lambda} = c d \log\left(1 + \frac{(2c_2 + c_3 c_{\theta}) \Psi(\bm{Z})}{d n_1\Lambda}\right) .\label{eq:epsilon_Lambda}
\end{align}

\end{theorem}
\noindent\textbf{Remark.} Theorem~\ref{theorem:train DP} ensures that the output model parameters $\bm{\Theta}_{priv}$ satisfy $(\epsilon, \delta)$-DP. Therefore, the privacy guarantee of GCON is independent of the optimization algorithm, such as SGD~\cite{xu2024stochastic} or Adam~\cite{kingma2014adam}, and the optimization steps. This offers GCON an advantage over traditional algorithms like DP-SGD because GCON does not need to compose the privacy cost of each optimization step, whereas traditional algorithms require an additional privacy budget for each optimization step.

We summarize the descriptions of notations used in Theorem~\ref{theorem:train DP} in Table~\ref{tab:Descriptions of notations}. 
$\Lambda'$ in Eq.~(\ref{eq:Lambda_prime}) is the coefficient of the quadratic perturbation term $\Lambda'\|\bm{\Theta}\|_F^2$. $\beta$ in Eq.~(\ref{eq:beta}) controls the distribution of $\bm{B}$ by Eq.~(\ref{eq:norm b distribution}), where $\bm{B}$ is the coefficient of the linear perturbation term $\bm{B}\odot\bm{\Theta}$.
In the input, $l(x;y)$ is the basic loss function introduced in Section~\ref{sec:GCON loss}, e.g., Eq.~(\ref{eq:mlsm loss}) \& (\ref{eq:ph loss}). 
$\omega$ is a hyperparameter dividing the privacy budget $\epsilon$ for computing the coefficients $\bm{B}$ and $\Lambda'$ of the two perturbation terms $\bm{B}\odot\bm{\Theta}$ and $\Lambda'\|\bm{\Theta}\|_F^2$ in Eq.~(\ref{eq:loss_priv}). The detailed usage of $\omega$ is presented in Appendix~\ref{app:proof of Theorem}. 

The parameters calculated in Eq.~(\ref{eq:supremum of l(x;y)}) - (\ref{eq:epsilon_Lambda}) are used to calculate the two core parameters $\beta$ and $\Lambda'$ in Eq.~(\ref{eq:Lambda_prime}) \& (\ref{eq:beta}). In Eq.~(\ref{eq:supremum of l(x;y)}), $c_1$, $c_2$, and $c_3$ are the supremums of $|l'|$, $|l''|$, and $|l'''|$ respectively, where $l'$, $l''$, and $l'''$ denote the first, second, and third order derivatives of $l(x;y)$ w.r.t. $x$. In Eq.~(\ref{eq:Psi_Z in Theorem}), $\Psi(\bm{Z})$ is the sensitivity of $\bm{Z}$ and we elaborate the analysis of it in Section~\ref{sec:Sensitivity Analysis}. In Eq.~(\ref{eq:Lambda constraint}), we set a lower bound for $\Lambda$ to facilitate subsequent analysis and calculations. In Eq.~(\ref{eq:find minimal c_sf}), (\ref{eq:c_theta}), \& (\ref{eq:epsilon_Lambda}), $c_{sf}$, $c_{\theta}$, and $\epsilon_{\Lambda}$ are intermediate calculation results for calculating $\beta$ and $\Lambda'$. 

\begin{table}[h]
\caption{Descriptions of notations}
\centering
\setlength\tabcolsep{1.5pt}
\begin{tabular}{c|c}
\toprule
\text{Notation} & \text{Description} \\
\midrule
$\Lambda'$ & \textbf{Coefficient of perturbation $\|\bm{\Theta}\|_F^2$ in Eq.~(\ref{eq:loss_priv})}  \\
$\beta$ & \textbf{Parameter for sampling $\bm{B}$ in Eq.~(\ref{eq:loss_priv}) \& (\ref{eq:norm b distribution})}  \\
$c$ & \text{Number of classes of nodes}  \\
$d$ & \text{Dimension of node features}  \\
$\Lambda$ & \text{Coefficient of regularization $\|\bm{\Theta}\|$ in Eq.~(\ref{eq:L_Lambda})}  \\
$l(x;y)$ & \text{Basic loss function, e.g., Eq.~(\ref{eq:mlsm loss}) \& (\ref{eq:ph loss})}  \\
$\omega$ & \text{Budget divider for the two perturbation terms}  \\
$\alpha$ & \text{Restart probability in Eq.~(\ref{eq:PPR matrix}) \& (\ref{eq:APPR matrix})}  \\
$\{m_i\}_{i=1}^s$ & \text{A series of propagation steps in Eq.~(\ref{eq:concatenate Z})}  \\
$c_1$, $c_2$, $c_3$ & \text{Supremums of $|l'|$,$|l''|$, $|l'''|$}  \\
$\Psi(\bm{Z})$ & \text{Sensitivity of $\bm{Z}$}  \\
$c_{sf}$, $c_{\theta}$, $\epsilon_{\Lambda}$ & \text{Intermediate calculation results for $\Lambda'$ and $\beta$}  \\
\bottomrule
\end{tabular}
\label{tab:Descriptions of notations}
\end{table}

\subsection{Analysis Overview}\label{sec:Analysis Overview}
Our proof has two primary stages: (i) quantifying the sensitivities of the PPR/APPR aggregate matrix $\bm{Z}_m=\bm{R}_m\bm{X}$ (Eq. (\ref{eq:Z_m})) and the concatenate feature matrix $\bm{Z}=\frac{1}{s}(\bm{Z}_{m_1} \oplus \cdots \oplus \bm{Z}_{m_s})$ (Eq. (\ref{eq:concatenate Z})), and (ii) based on the findings from (i), quantifying the correlation between the noise $\bm{B}$ and the optimally derived network parameters $\bm{\Theta}_{priv}$.

For stage (i), our analysis starts with the sensitivity of the normalized adjacency matrix $\Tilde{\bm{A}}$, and progressively extends to $\Tilde{\bm{A}}^m,\bm{R}_m$, $\bm{Z}_m$, and $\bm{Z}$. As $\bm{A}$ is usually enforced to be symmetric in GCNs, adding or removing one edge alters two elements in $\Tilde{\bm{A}}$ ($\bm{A}$ is usually enforced symmetric in practice). Leveraging several pivotal properties of $\Tilde{\bm{A}}^m$ presented in Lemma~\ref{lemma:AR properties} (whose proof is presented in Appendix~\ref{app:AR properties}), we employ mathematical induction to bound the sensitivities of $\Tilde{\bm{A}},\Tilde{\bm{A}}^2,\cdots,\Tilde{\bm{A}}^m$, and their linear combinations $\bm{R}_m$ (we analyze the special case $\bm{R}_{\infty} $ separately). Given that the feature matrix $\bm{X}$ remains constant across edge-level neighboring graphs, the sensitivity of $\bm{Z}_m$ is inferred through standard matrix multiplications. Lastly, by summing the sensitivities of $\{\bm{Z}_{m_i}\}_{i=1}^s$, we obtain the sensitivity of $\bm{Z}$, which is evaluated by $\max(\sum_i^n \|\bm{z}_i - \bm{z}_i'\|_2)$, as defined in Definition~\ref{def:psi_Z}.
\begin{lemma}\label{lemma:AR properties}
Let the elements in $\Tilde{\bm{A}}$ satisfy the following constraints
\begin{displaymath}
    \Tilde{\bm{A}}_{ij}=\left\{
    \begin{array}{lcl}
    0 & & i\neq j \,\&\, \bm{A}_{ij}=0 \\
    \min(\frac{1}{k_i + 1}, p) & & i\neq j\,\&\,\bm{A}_{ij}=1\\
    1 - \sum_{u\neq i} \Tilde{\bm{A}}_{iu}& & i=j\\
    \end{array} \right. 
\end{displaymath}
where $k_i$ is the degree of node $i$ and $p\leq \frac{1}{2}$ is a bound on the off-diagonal elements in $\Tilde{\bm{A}}$. Especially, when $p=\frac{1}{2}$, it is equivalent to the original normalization (in Section~\ref{sec:GCON Propagation}) for $\Tilde{\bm{A}}$ without any additional bound $p$.
Then for any integer $m\geq 1$, we have
\begin{itemize}
    \item Every entry of $\Tilde{\bm{A}}^m$, $\bm{R}_m$, or $\bm{R}_{\infty}$ is non-negative.
    \item The sum of each row of $\Tilde{\bm{A}}^m$, $\bm{R}_m$, or $\bm{R}_{\infty}$ is 1.
    \item The sum of column $i$ of $\Tilde{\bm{A}}^m$, $\bm{R}_m$, or $\bm{R}_{\infty}$ is $\leq \max((k_i+1)p, 1)$.
\end{itemize}
\end{lemma}

\begin{definition}[Sensitivity metric]\label{def:psi_Z}
    Given a graph dataset $D$, let $D'$ be any of its edge-level neighboring dataset. Let $\bm{Z}$ and $\bm{Z}'$ be two matrices computed in the same way (e.g., by Eq. (\ref{eq:Z_m}) or (\ref{eq:concatenate Z})) on $D$ and $D'$, respectively. Let $\bm{z}_i$ and $\bm{z}'_i$ denote the $i$-th rows of $\bm{Z}$ and $\bm{Z}'$, respectively. We define a metric $\psi(\bm{Z})$ to evaluate the difference between $\bm{Z}$ and $\bm{Z}'$ as follows.
    \begin{equation*}
    \psi(\bm{Z}) \triangleq \sum_{i=1}^{n}\|\bm{z}_i' - \bm{z}_i\|_2.
    \end{equation*}
    Then the sensitivity of $\bm{Z}$ is
    \begin{equation*}
    \Psi(\bm{Z}) = \max_{D,D'} \psi(\bm{Z}).
    \end{equation*}
\end{definition}

For stage (ii), it has been demonstrated in DPERM~\cite{chaudhuri2011DPERM} that, given graph datasets $D$ and $D'$, the ratio of the densities of two optimal network parameters $\bm{\Theta}_{priv}$ and $\bm{\Theta}_{priv}'$ is proportional to the ratio of densities $\frac{\mu(\bm{B} | D)}{\mu(\bm{B}' | D')}$ of two corresponding noises $\bm{B}$ and $\bm{B}'$, and inversely proportional to the ratio of the determinants of the respective Jacobian matrices $\bm{J}(\bm{\Theta}_{priv}\rightarrow \bm{B}|D)$ and $ \bm{J}(\bm{\Theta}_{priv}'\rightarrow \bm{B}'|D')$. Our sampling method, detailed in Algorithm~\ref{alg:sampling}, ensures that each column $\bm{b}_j$ of $\bm{B}$ is independent of the other columns. Similarly, the derivative of each column $\bm{\theta}_j$ of $\bm{\Theta}_{priv}$ is independent of those of the other columns in the context of the perturbed loss function $L_{priv}$ as specified in Section~\ref{sec:GCON pert}. Hence, we focus on examining the ratios of densities of noises $\frac{\mu(\bm{b}_j | D)}{\mu(\bm{b}_j' | D')}$ and the ratios of determinants of the Jacobian matrices $|\frac{\det(\bm{J}(\bm{\theta}_j \mapsto \bm{b}_j' | D')}{\det(\bm{J}(\bm{\theta}_j \mapsto \bm{b}_j | D)}|$ (denoted as $|\frac{\det(\bm{J}')}{\det(\bm{J})}|$), for all $j\in [1,c]$.

The ratio $\frac{\mu(\bm{b}_j | D)}{\mu(\bm{b}_j' | D')}$ can be bounded by a meticulously chosen parameter $\beta$. For the latter, the ratio of determinants of two matrices satisfies $|\frac{\det(\bm{J}')}{\det(\bm{J})}|=|\det(\bm{J}' \bm{J}^{-1})|$. Given that the determinant of a matrix is equivalent to the product of its singular values, our analysis shifts focus to examining the singular values of $\bm{J}' \bm{J}^{-1}$. By applying the Courant–Fischer min-max theorem~\cite{courant2008methods}, we bound each singular value of $\bm{J}' \bm{J}^{-1}$ through the product of the respective singular values of $\bm{J}'$ and $\bm{J}^{-1}$. The singular values of $\bm{J}^{-1}$ are the reciprocals of the singular values of $\bm{J}$. The mapping $\bm{\Theta}_{priv}\rightarrow \bm{B}|D$ and $ \bm{\Theta}_{priv}'\rightarrow \bm{B}'|D'$ are constituted by the loss function, which aggregates the loss for features $\bm{z}_i$ and $\bm{z}_i'$ respectively. 
Corollary 3.4.3 in \cite{horn1991topics} indicates that the sum of singular values of a matrix resulting from addition is less than the sum of the singular values of each individual matrix. Consequently, we decompose the Jacobian matrix into a summation of $n_1$ matrices, each associated with $\bm{z}_i$ or $\bm{z}_i'$, and analyze their respective singular values utilizing Corollary 3.4.3 in \cite{horn1991topics}. The discrepancy in singular values of matrices associated with $\bm{z}_i$ and $\bm{z}_i'$ is quantifiable through $\|\bm{z}_i - \bm{z}_i'\|_2$. Ultimately, we link the sum of these singular values to the previously established sensitivity of $\bm{Z}$, enabling us to derive $|\frac{\det(\bm{J}')}{\det(\bm{J})}|$. 

In the following, we provide a detailed analysis of the sensitivities of $\bm{Z}_m$ and $\bm{Z}$ in Section~\ref{sec:Sensitivity Analysis}, which forms a foundational component of the theorem's proof. 

\subsection{Sensitivities of \texorpdfstring{\(\bm{Z}_m\)}{} and \texorpdfstring{\(\bm{Z}\)}{}}\label{sec:Sensitivity Analysis}
We define the sensitivities of $\bm{Z}_m$ and $\bm{Z}$ by Definition~\ref{def:psi_Z}, denoted as $\Psi(\bm{Z}_m)$ and $\Psi(\bm{Z})$, respectively. Since $\bm{Z}$ is simply a concatenation of $\bm{Z}_m$, the difficulty and focus of sensitivity analysis lie on $\bm{Z}_m$. $\Psi(\bm{Z}_m)$ and $\Psi(\bm{Z})$ are established in Lemma~\ref{lemma:Psi_Z of PPR and APPR}. Lemma~\ref{lemma:Psi_Z of PPR and APPR} elucidates the correlations of the sensitivity of $\bm{Z}_m$ with two key factors: the propagation step $m$ and the restart probability $\alpha$. The rationale behind these correlations is as follows. Firstly, a larger propagation step $m$, results in each node disseminating its features to a more expansive neighborhood, thereby amplifying the sensitivity.
Secondly, a smaller value of $\alpha$ leads to each step of the propagation matrix incorporating more information from the edges (indicated by $(1-\alpha)\cdot \Tilde{\bm{A}}$). This elevates the influence of individual edges within the graph, thereby increasing the sensitivity.

These correlations provide valuable insights for managing sensitivity within the PPR/APPR schemes, which are advantageous in practice. These insights reveal a trade-off: balancing the amount of sensitivity, which directly influences the amount of additional noise needed for privacy protection, against the efficacy of the propagation matrix. For instance, in practical scenarios constrained by limited privacy budgets, opting for a smaller propagation step $m$ or a larger restart probability $\alpha$ can reduce the sensitivity. This reduction, in turn, diminishes the additional noise, thereby offering a pathway to enhance the model's performance.

\begin{lemma}\label{lemma:Psi_Z of PPR and APPR}
For $\bm{Z}_m$ computed by Eq. (\ref{eq:Z_m}), we have
\begin{equation}\label{eq:Psi Z_m}
    \Psi(\bm{Z}_m) = \frac{2(1-\alpha)}{\alpha}[1-(1-\alpha)^m],
\end{equation}
Specifically, 
\begin{equation*}
    \Psi(\bm{Z}_\infty) = \lim_{m\to\infty} \frac{2(1-\alpha)}{\alpha}[1-(1-\alpha)^m] = \frac{2(1-\alpha)}{\alpha} .
\end{equation*}
For $\bm{Z}$ computed by Eq. (\ref{eq:concatenate Z}), we have
\begin{equation}\label{eq:Psi Z}
    \Psi(\bm{Z}) = \frac{1}{s}\sum_{i}^{s} \Psi(\bm{Z}_{m_i}).
\end{equation}
\end{lemma}

The proof of Lemma~\ref{lemma:Psi_Z of PPR and APPR}, detailed in Appendix~\ref{app:Proof of Psi_Z}, progressively establishes bounds for $\Tilde{\bm{A}}^m$, $\bm{R}_m$, $\bm{Z}_m$, and $\bm{Z}$. As defined in Definition~\ref{def:psi_Z}, bounding the sensitivity involves constraining the row-wise differences $\|\bm{z}_i' - \bm{z}_i\|_2$ (corresponding to bounding the sum of each row) and subsequently aggregating these differences across all rows $i=1,2,\ldots,n$ (corresponding to bounding the sum of each column). In terms of rows, we observe that the sum of each row in $\Tilde{\bm{A}}^m$ and $\bm{R}_m$ consistently equals 1, a consequence of the normalization process ($\Tilde{\bm{A}} = \bm{D}^{-1}(\bm{A} + \bm{I})$) and linear combination operations described in Section~\ref{sec:GCON Propagation}. However, bounding the columns of $\Tilde{\bm{A}}^m$ and $\bm{R}_m$ is rather complex due to the intricacies introduced during propagation.

To address this difficulty, we leverage Lemma~\ref{lemma:AR properties}, which serves as a cornerstone in the proof of Lemma~\ref{lemma:Psi_Z of PPR and APPR}. 
Lemma~\ref{lemma:AR properties} establishes that the sum of the $j$-th column of $\Tilde{\bm{A}}^m$ or $\bm{R}_m$ is intrinsically linked to the degree of node $j$, irrespective of the number of propagation steps $m$. Remarkably, this relationship holds even when $\Tilde{\bm{A}}^m$ undergoes an element-wise clipping, a technique frequently employed in DP algorithms. The implications of Lemma~\ref{lemma:AR properties} extend beyond our analysis, offering valuable insights that can benefit future research in the realm of differentially private algorithms, especially in contexts where artificial clippings are commonplace.

\section{Experiments}\label{sec:Experiments}
\subsection{Settings}
\noindent\textbf{Datasets.}
Following previous work~\cite{ppnp, kolluri2022lpgnet, huang2024optimizing}, we evaluate the methods on 4 real-world datasets with varied sizes and homophily ratios, as defined in Definition~\ref{def:homo ratio}. The homophily ratio indicates the tendency for nodes with identical labels to be connected. Specifically, we consider three homophily graphs (Cora-ML, CiteSeer, and PubMed) and one heterophily graph (Actor).
Key statistics of these datasets and the train-test sets split are summarized in Appendix~\ref{app:Datasets Statistics}. The training and testing are conducted on the same graph and both are considered private. We evaluate the methods' classification accuracy on the datasets by the micro-averaged F1 score (Micro F1), a metric that combines the precision and recall scores of the model.

\vspace{2pt}
\noindent\textbf{Competitors.} We evaluate 7 baselines: {\sf LPGNet}~\cite{kolluri2022lpgnet}, {\sf DPGCN}~\cite{wu2022DPGCN}, {\sf GAP}~\cite{sajadmanesh2023gap}, {\sf ProGAP}~\cite{sajadmanesh2023progap}, {\sf DP-SGD}~\cite{abadi2016dpsgd}, {\sf GCN (Non-DP)}, and {\sf MLP}. For the non-private baseline {\sf GCN (non-DP)}, we follow \cite{kipf2016semi, chien2020adaptive} as referenced in \cite{huang2024optimizing, ma2021homophily}. {\sf MLP} denotes a simple multi-layer perceptron model that does not utilize graph edges and, therefore, inherently satisfies edge DP for any privacy budget. {\sf GAP} and {\sf ProGAP} refer to the GAP-EDP~\cite{sajadmanesh2023gap} and ProGAP-EDP~\cite{sajadmanesh2023progap} algorithms, both of which ensure edge-level DP. All baselines except {\sf GCN (Non-DP)} are evaluated under the same edge-level DP settings in our experiments.

\begin{figure*}[t]
\centering
\begin{tikzpicture}
    \begin{customlegend}[legend columns=-1,legend style={align=left,font=\small,draw=none,/tikz/every even column/.append style={column sep=4mm}},legend entries={{\sf GCON}, {\sf DP-SGD}, {\sf DPGCN}, {\sf LPGNet}, {\sf GAP}, {\sf ProGAP}, {\sf MLP}, {\sf GCN (non-DP)}}]
    \addlegendimage{line width=1.4pt, blue, mark=square*, mark size=2pt, mark options={solid,fill opacity=0}},
    \addlegendimage{semithick, springgreen, mark=triangle*, mark size=2pt, mark options={solid,rotate=180}},
    \addlegendimage{semithick, goldenrod1, mark=|, mark size=2pt, mark options={solid}},
    \addlegendimage{semithick, darkorchid, mark=+, mark size=2pt, mark options={solid}},
    \addlegendimage{semithick, darkorange1, mark=x, mark size=2pt, mark options={solid}},
    \addlegendimage{semithick, darkgreen, mark=asterisk, mark size=2pt, mark options={solid,fill opacity=0}},
    \addlegendimage{semithick, brick, mark=triangle*, mark size=2pt, mark options={solid,fill opacity=0}},
    \addlegendimage{semithick, hotpink255129192, mark=triangle, mark size=2pt, mark options={solid,rotate=180,fill opacity=0}},
    \end{customlegend}
\end{tikzpicture}
\\
\begin{subfigure}[h]{0.24\textwidth}
\includegraphics[width=\textwidth]{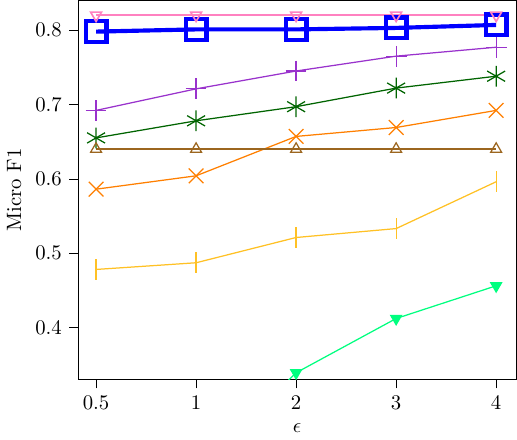}
\caption{Cora-ML}
\label{fig:Cora-ML acc}
\end{subfigure}
\begin{subfigure}[h]{0.24\textwidth}
\includegraphics[width=\textwidth]{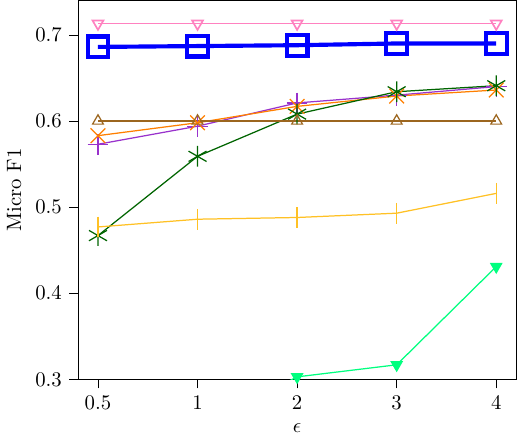}
\caption{CiteSeer}
\label{fig:CiteSeer acc}
\end{subfigure}
\begin{subfigure}[h]{0.24\textwidth}
\includegraphics[width=\textwidth]{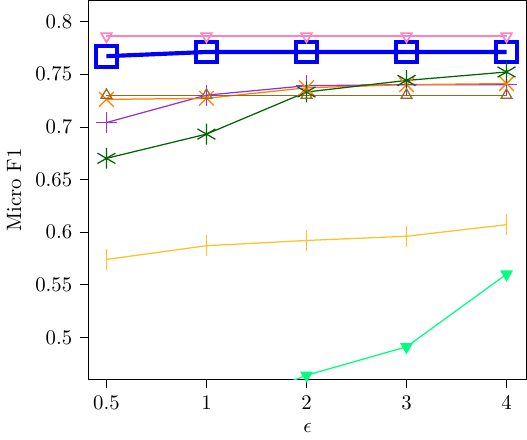}
\caption{PubMed}
\label{fig:PubMed acc}
\end{subfigure}
\begin{subfigure}[h]{0.24\textwidth}
\includegraphics[width=\textwidth]{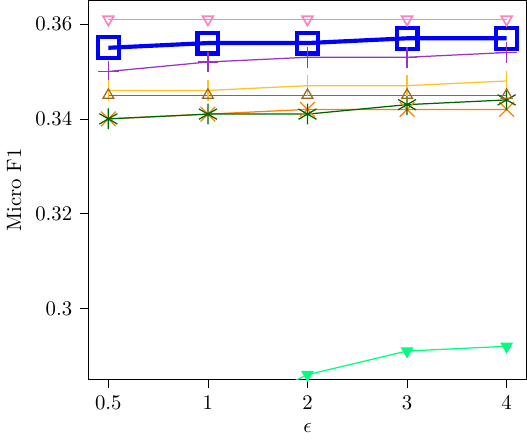}
\caption{Actor}
\label{fig:Actor acc}
\end{subfigure}
\caption{Model performance  (micro F1 score) versus privacy constraints on different datasets. GCN (non-DP) represents the target performance (i.e., an upper bound) for all DP algorithms.}
\label{fig:algs acc}
\end{figure*}

\vspace{2pt}
\noindent\textbf{Hyperparameters.}
For the privacy constraints, we vary $\epsilon$ across $\{0.5, 1, 2, 3, 4\}$ and set $\delta=\frac{1}{|E|}$ ($|E|$ is the number of edges). 
The detailed hyperparameter settings for our method are summarized in Appendix~\ref{app:Hyperparameters}.

\begin{figure*}[t]
\centering
\begin{tikzpicture}
    \begin{customlegend}[legend columns=-1,legend style={align=left,font=\small,draw=none,/tikz/every even column/.append style={column sep=5mm}},legend entries={$\alpha=0.8$, $\alpha=0.6$, $\alpha=0.4$, $\alpha=0.2$ }]
    \addlegendimage{semithick, blue1, mark=square*, mark size=2pt, mark options={solid,fill opacity=0}},
    \addlegendimage{semithick, red1, mark=*, mark size=2pt, mark options={solid,fill opacity=0}},
    \addlegendimage{semithick, darkorange1, mark=x, mark size=2pt, mark options={solid}},
    \addlegendimage{semithick, darkgreen, mark=asterisk, mark size=2pt, mark options={solid,fill opacity=0}},
    \end{customlegend}
\end{tikzpicture}
\\
\vspace{1mm}
\begin{subfigure}[h]{0.32\textwidth}
\includegraphics[width=\textwidth]{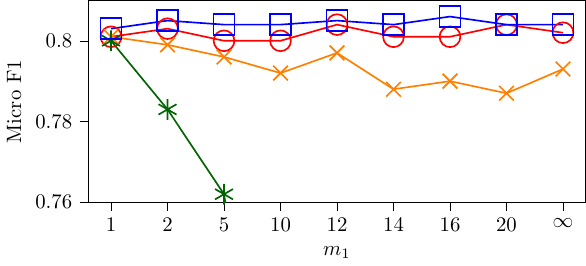}
\caption{Cora-ML}
\label{fig:Cora-ML m}
\end{subfigure}
\begin{subfigure}[h]{0.32\textwidth}
\includegraphics[width=\textwidth]{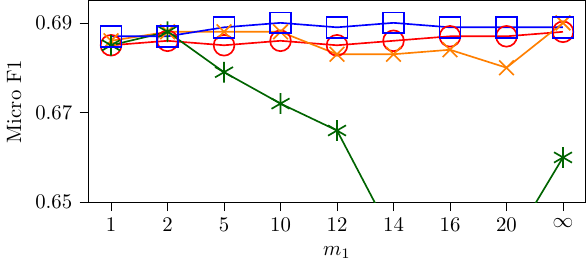}
\caption{CiteSeer}
\label{fig:CiteSeer m}
\end{subfigure}
\begin{subfigure}[h]{0.32\textwidth}
\includegraphics[width=\textwidth]{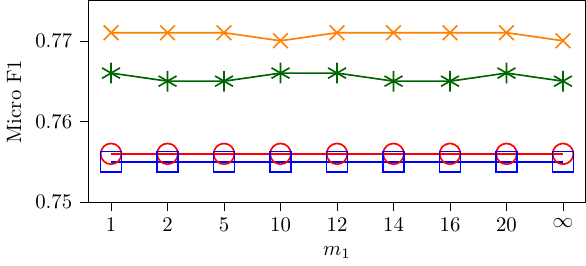}
\caption{PubMed}
\label{fig:PubMed m}
\end{subfigure}
\caption{Effect of the propagation step in {\sf GCON} on private test graph under $\epsilon=4$.}
\label{fig:Ablation m}
\end{figure*}
\begin{figure*}[t]
\centering
\begin{tikzpicture}
    \begin{customlegend}[legend columns=-1,legend style={align=left,font=\small,draw=none,/tikz/every even column/.append style={column sep=5mm}},legend entries={$\alpha=0.8$, $\alpha=0.6$, $\alpha=0.4$, $\alpha=0.2$ }]
    \addlegendimage{semithick, blue1, mark=square*, mark size=2pt, mark options={solid,fill opacity=0}},
    \addlegendimage{semithick, red1, mark=*, mark size=2pt, mark options={solid,fill opacity=0}},
    \addlegendimage{semithick, darkorange1, mark=x, mark size=2pt, mark options={solid}},
    \addlegendimage{semithick, darkgreen, mark=asterisk, mark size=2pt, mark options={solid,fill opacity=0}},
    \end{customlegend}
\end{tikzpicture}
\\
\begin{subfigure}[h]{0.32\textwidth}
\includegraphics[width=\textwidth]{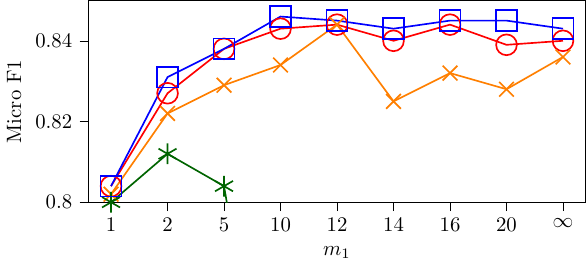}
\caption{Cora-ML}
\label{fig:Cora-ML m nondp}
\end{subfigure}
\begin{subfigure}[h]{0.32\textwidth}
\includegraphics[width=\textwidth]{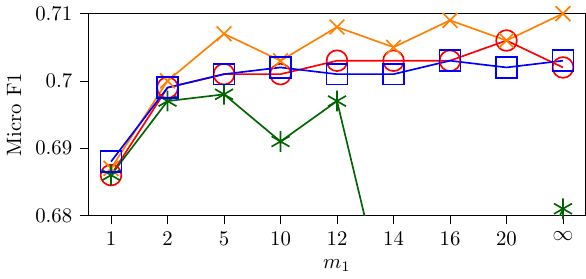}
\caption{CiteSeer}
\label{fig:CiteSeer m nondp}
\end{subfigure}
\begin{subfigure}[h]{0.32\textwidth}
\includegraphics[width=\textwidth]{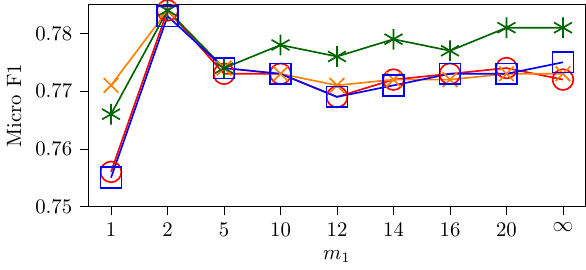}
\caption{PubMed}
\label{fig:PubMed m nondp}
\end{subfigure}
\caption{Effect of the propagation step in {\sf GCON} on public test graph under $\epsilon=4$.}
\label{fig:Ablation m nondp}
\end{figure*}
\begin{figure}[htbp]
\centering
\begin{tikzpicture}
    \begin{customlegend}[legend columns=-1,legend style={align=left,font=\small,draw=none,/tikz/every even column/.append style={column sep=5mm}},legend entries={$\alpha=0.8$, $\alpha=0.6$, $\alpha=0.4$, $\alpha=0.2$ }]
    \addlegendimage{semithick, blue1, mark=square*, mark size=2pt, mark options={solid,fill opacity=0}},
    \addlegendimage{semithick, red1, mark=*, mark size=2pt, mark options={solid,fill opacity=0}},
    \addlegendimage{semithick, darkorange1, mark=x, mark size=2pt, mark options={solid}},
    \addlegendimage{semithick, darkgreen, mark=asterisk, mark size=2pt, mark options={solid,fill opacity=0}},
    \end{customlegend}
\end{tikzpicture}
\\
\begin{subfigure}[h]{0.15\textwidth}
\includegraphics[width=\textwidth]{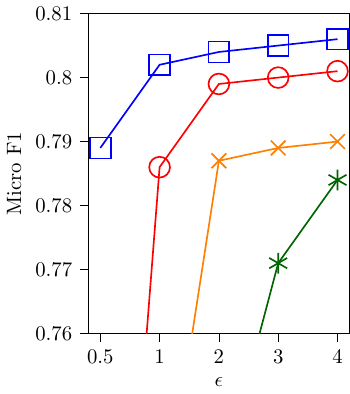}
\caption{Cora-ML}
\label{fig:Cora-ML a}
\end{subfigure}
\begin{subfigure}[h]{0.15\textwidth}
\includegraphics[width=\textwidth]{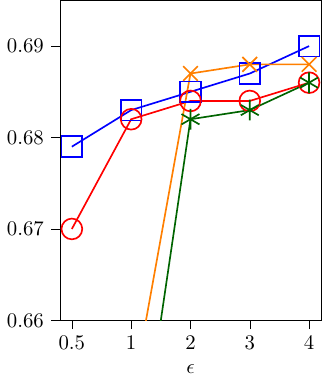}
\caption{CiteSeer}
\label{fig:CiteSeer a}
\end{subfigure}
\begin{subfigure}[h]{0.15\textwidth}
\includegraphics[width=\textwidth]{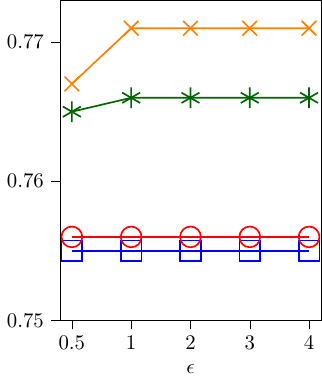}
\caption{PubMed}
\label{fig:PubMed a}
\end{subfigure}
\caption{Effect of the restart probability $\alpha$ in {\sf GCON}.}
\label{fig:Ablation alpha}
\end{figure}

\subsection{Results}
We repeat each experimental setup for 10 independent runs and report the average Micro F1 scores of the methods across varying $\epsilon \in \{0.5, 1, 2, 3, 4\}$. As depicted in Figure~\ref{fig:algs acc}, {\sf GCON} consistently outperforms the competitors in all configurations, indicating that the {\sf GCON} model is not only more accurate but also demonstrates greater versatility across various privacy budgets and datasets.

Particularly, {\sf GCON} achieves a higher performance improvement on homophily datasets Cora-ML, CiteSeer, and PubMed. In these datasets, nodes with identical labels are more likely to be connected. Recall from Sections~\ref{sec:GCN preliminary} and \ref{sec:SGC} that both the standard GCN and the PPNP/APPNP propagation schemes aggregate neighbor nodes' information through their edges, which are represented by the adjacency matrix ($\Tilde{\bm{A}}$) or the propagation matrix. This aggregation (e.g., $\Tilde{\bm{A}}\bm{X}$) on homophily graphs facilitates the learning of node characteristics under each category, thereby improving classification accuracy. However, competing methods add noise directly to the adjacency matrix, disrupting their aggregation process and learning of classification due to interference from other categories. Thus, their performance deteriorates when the privacy budget $\epsilon$ is small due to substantial injected noise and improves as $\epsilon$ increases. In contrast, {\sf GCON} maintains an unaltered aggregation process and only introduces noise to $\bm{\Theta} \in \mathbb{R}^{d\times c}$ instead of overwhelming noise to $\Tilde{\bm{A}} \in \mathbb{R}^{N\times N}$ or $\Tilde{\bm{A}}\bm{X} \in \mathbb{R}^{N\times d}$ ($N\gg c,d$). Therefore, {\sf GCON} can withstand variations in privacy budget and maintain a superior performance close to the {\sf GCN (Non-DP)} baseline. On the Actor dataset, the performance gap is smaller for different methods, since on this dataset, nodes with different labels are interconnected, and there is no clear pattern of nodes and their neighborhoods. This reduces the benefits of an unaltered neighbor aggregation; nonetheless, {\sf GCON} outperforms all other DP competitors.

\vspace{2pt}
\noindent\textbf{Effect of the propagation step $m_1$.} In what follows, we study the influence of the propagation step $m_1$ in {\sf GCON} under $\epsilon=4$. We plot the performance under different $\alpha \in \{0.2,0.4,0.6,0.8\}$ in Figure~\ref{fig:Ablation m}. As $m_1$ increases, the $\alpha=0.2$ curve rapidly declines and the $\alpha=0.4$ curve slightly declines, while the other two curves remain stable or slightly increase. The decline of the $\alpha=0.2$ and $\alpha=0.4$ curves is due to their higher sensitivity (as indicated in Lemma~\ref{lemma:Psi_Z of PPR and APPR}), which results in more injection noise. In contrast, the sensitivity and injected noise for the other two curves are less affected by $m_1$. 

Meanwhile, with $\alpha=0.6$ and $\alpha=0.8$, the graph convolution process can absorb richer information from a larger neighborhood, allowing the model to better learn and classify node features, thereby more than compensating for the increased noise. Additionally, to further evaluate the effectiveness of the trained model, we present the results of {\sf GCON} with non-private inference in Figure~\ref{fig:Ablation m nondp}, following \cite{daigavane2021node, chien2024dpdgc, zhang2024dpar}, where the test graph is considered public and the edges in the private training set are all excluded (i.e., removing the constraints of private inference). Figure~\ref{fig:Ablation m nondp} demonstrates that the model's performance continues to improve when $m_1 \leq 10$ and stabilizes after $m_1 > 10$. This is because, when $m_1$ increases beyond 10, the expanded range of graph convolution leads to diminishing returns in classification ability and the model performance is further impacted by increased sensitivity and injection noise. Thus, summarizing Figures~\ref{fig:Ablation m} \& \ref{fig:Ablation m nondp}, a propagation step of $10$ or fewer is sufficient to train satisfactory models, whether tested on private or public graphs.

\vspace{2pt}
\noindent\textbf{Effect of the Restart Probability $\alpha$.} In Figure~\ref{fig:Ablation alpha}, we report the accuracy of {\sf GCON} on Cora-ML, CiteSeer, and PubMed with varying $\alpha \in \{0.2, 0.4, 0.6, 0.8\}$ and $m_1=2$ under $\epsilon \in \{0.5, 1, 2, 3, 4\}$. On Cora-ML and CiteSeer, $\alpha=0.8$ is the best choice to achieve better model performance, while $\alpha=0.4$ is the best on PubMed. However, at $\alpha=0.2$, the performance is relatively poor, especially under a restrictive privacy budget ($\epsilon \leq 1$). The reason is that lower $\alpha$ values lead to higher sensitivity of $\bm{Z}_{m_1}$, as indicated in Lemma~\ref{lemma:Psi_Z of PPR and APPR}, resulting in a larger amount of injected noise that hinders model performance. Moreover, even with an ample privacy budget ($\epsilon=3,4$), $\alpha=0.2$ offers little improvement in model performance. Hence, a restart probability $\alpha$ of $\geq 0.4$ proves to be more broadly useful in practice.

Summarizing the experimental results, {\sf GCON} significantly outperforms its competitors in all settings. Thus, it is the method of choice in practical applications. Among the configurations of {\sf GCON}, employing the APPR scheme with $m\leq 10$ and a larger $\alpha \geq 0.4$ is generally a good choice.
\section{Related Work}\label{sec:Related Work}
Graph Neural Networks (GNNs) have emerged as a powerful tool for analyzing graph-structured data, which typically consists of nodes interconnected by complex edges representing relationships and dependencies. GNNs encompass various architectures, including Graph Convolutional Networks (GCNs)~\cite{kipf2016semi}, Graph Recurrent Neural Networks (GRNNs)~\cite{hu2020stochastic}, and Graph Autoencoders (GAEs), with GCNs being particularly prevalent in node classification tasks due to their effectiveness and wide adoption in both research~\cite{chen2024user} and practical applications~\cite{lu2024pr}. For a comprehensive overview of GNN variants and their applications, readers are referred to Wu et al.~\cite{wu2020comprehensive}.

Privacy concerns also arise when using GNNs to analyze sensitive graph data, such as social networks~\cite{xia2023disentangled}, medical knowledge bases~\cite{vretinaris2021medical}, and patient relationships~\cite{schrodt2020patient}. There have been several successful privacy attacks. To briefly mention a few (non-exhaustive), Zhang et al.~\cite{zhang2022inference} investigate privacy leakages in graph properties and propose several inference attacks; Meng et al. present infiltration schemes for inferring neighboring links~\cite{meng2023devil}; and He et al. develop a suite of black-box attacks targeting edge existence~\cite{he2021stealing}.

Differential Privacy (DP), whose application spans a variety of algorithms\cite{dwork2014algorithmic, bernau2021quantifying, cao2017quantifying}, offers a rigorous framework for mitigating privacy risks in graph data analytics~\cite{karwa2011private, jorgensen2016publishing, zhang2015private, ye2020lf}. DP is often achieved by injecting noise into sensitive information. Here noise can be injected into various components of an algorithm, including its input~\cite{yang2023fortifying}, intermediate results~\cite{mohapatra2023differentially}, and the output~\cite{liu2024cargo}. Besides, objective function~\cite{chaudhuri2011DPERM}, which perturbs the objective function (of a strongly convex optimization problem) is also widely used in machine learning applications such as logistic regression. Training deep neural networks, on the other hand, often requires additional types of machinery such as privacy composition and data subsampling~\cite{abadi2016dpsgd,mironov2019renyi}. The aforementioned well-established techniques, however, do not automatically apply to training differentially private GCNs. The reason is that private data records (i.e., edges in the context of edge DP) in the sensitive input dataset of GCN training algorithms are usually correlated, as opposed to being independent as in other applications such as statistical machine learning, computer vision, and natural language processing~\cite{chaudhuri2011DPERM,abadi2016dpsgd,de2022unlocking,rust2023differential}.

Recent research on the combination of DP and GNNs primarily considers two DP definitions: edge DP and node DP. Edge DP\cite{liu2022collecting, ran2024differentially, liu2024edge} prevents inferences about individual edges, while node DP protects all edges and the node itself~\cite{day2016publishing}. In this work, we focus on Edge-DP. Existing Edge-DP approaches in GNNs typically perturb the adjacency or aggregate feature matrix~\cite{wu2022DPGCN, kolluri2022lpgnet, sajadmanesh2023gap,sajadmanesh2023progap}. For instance, {\sf DP-GCN}~\cite{wu2022DPGCN} adds noise directly to the adjacency matrix, leading to significant distortions. {\sf GAP}~\cite{sajadmanesh2023gap} uses the original adjacency matrix to aggregate node features first and then adds noise to the aggregate features. Building upon {\sf GAP}, {\sf ProGAP}~\cite{sajadmanesh2023progap} also employs an iterative process for feature aggregation using the adjacency matrix. The critical distinction lies in the handling of the perturbed aggregate features: in each iteration, these are processed through an MLP, and the resultant outputs across iterations are concatenated. {\sf LPGNet}~\cite{kolluri2022lpgnet} breaks down graphs into independent nodes, and aggregates and then perturbs node features. Our work diverges from these methods by perturbing only the objective function, retaining the integrity of the graph convolutions, and enhancing classification accuracy.
We also note that there is growing interest in  GNNs with node DP~\cite{daigavane2021node} and decentralized graph data~\cite{sajadmanesh2021locally, he2024butterfly}, which is orthogonal to the focus of this work. 
\section{Conclusion}\label{sec:Conclusion}
In this paper, we propose {\sf GCON}, a novel algorithm for publishing GCNs with edge differential privacy. {\sf GCON} is distinctive in its use of objective perturbation, which is challenging in the context of GCNs due to high sensitivity and non-convexity of traditional GCN architectures. 
The efficacy of {\sf GCON} is primarily attributed to its unaltered graph convolution operations in the GCN. The privacy guarantee of {\sf GCON} is rigorously established through formal proofs. We conducted a comprehensive series of experiments on popular graph datasets, exploring {\sf GCON}'s performance under diverse privacy budget scenarios. The results from these experiments consistently demonstrate that it not only achieves significant improvements in accuracy over existing methods, but also exhibits robust generalization capabilities across various graphs.

\bibliographystyle{IEEEtran}
\bibliography{main}
\newpage

\appendix
\subsection{Algorithms}\label{app:Algorithms}
Algorithm~\ref{alg:sampling} shows how we sample the specific noise. Algorithm~\ref{alg:sampling} first samples the radius $a$ from the Erlang distribution (Eq. (\ref{eq:norm b distribution})) that has been implemented in many popular libraries~\cite{2020SciPy-NMeth}. To sample each direction in the hypersphere with equal probability, we then sample a vector $\bm{u}$ from a $d$-dimensional normal distribution and scale its length to $a$. The correctness of Algorithm~\ref{alg:sampling} is explained in Appendix~\ref{app:Gaussian rotational symmetry}.
\begin{algorithm}[h]\small
\caption{Noise sampling}
\label{alg:sampling}
\KwIn{Dimension $d$, distribution parameter $\beta$}
\KwOut{Vector $\bm{b_1}$ with dimension $d$}
Sample $a \in (0, +\infty)$ with a PDF shown in Eq. (\ref{eq:norm b distribution})\;\label{line:sample radius}
Sample $u_1, \dots, u_d$ independently according to the normal distribution $\mathcal{N}(0,1)$\;\label{line:sample direction begin}
\For{$i \in [1, d]$}
{
    $b_{1,i} \gets  a u_i / \sqrt{u_1^2 + \dots + u_d^2}$
}
$\bm{b_1} = (b_{1,1}, \dots, b_{1,d})$\;\label{line:sample direction end}
\Return $\bm{b_1}$\;
\end{algorithm}

Algorithm~\ref{alg:encoder} shows the workflow of our feature encoder. 
\begin{algorithm}[h]\small
\caption{Encoder($\bm{X}, \bm{Y}, d$)}
\label{alg:encoder}
\KwIn{labeled and unlabeled features $\bm{X}=\{\bm{X}_l,\bm{X}_{ul}\}$, $\bm{Y}=\{\bm{Y}_l,\bm{Y}_{ul}\}$, expected dimension $d$, a loss function $L_{mlp}$}
\KwOut{Encoded features and labels $\bm{X}=\{\bm{X}_l,\Tilde{\bm{X}}_{ul}\}$}
Initialize a MLP model $\mathcal{M}$ with parameters $\bm{W}_1$, a classification layer $\bm{W}_2$, and activation functions $H_{mlp},H$\;
$\Tilde{\bm{X}_l} = H_{mlp}\left(\mathcal{M}(\bm{X}_l)\right)$ \;
$\Tilde{\bm{Y}_l} =H\left(\Tilde{\bm{X}_l} \bm{W}_2\right)$ \;
$\bm{W}_1^{\star}, \bm{W}_2^{\star}=\argmin _{\bm{W}_1,\bm{W}_2} L_{mlp}(\Tilde{\bm{Y}_l}, \bm{Y}_l)$ \;
$\Tilde{\bm{X}} = H_{mlp}\left(\operatorname{MLP}(\bm{X} ; \bm{W}_1^{\star} )\right)$ \;
\Return $\bm{X}=\{\bm{X}_l,\Tilde{\bm{X}}_{ul}\}$\;
\end{algorithm}

Algorithm~\ref{alg:Inference} shows the inference procedure.
\begin{algorithm}[ht]\small
\caption{Inference}
\label{alg:Inference}
\KwIn{Optimized model $\bm{\Theta}_{priv}$, target dataset $D=\{V,E,\bm{A},\bm{X},\bm{Y}\}$, propagation steps $m_1, m_2, \cdots, m_s$, and $\alpha_I$. }
\KwOut{Label predictions $\hat{\bm{Y}}$.}
\If{private inference}
{\label{line:private inference}
    \For{$i \in [1,s]$}
    {$\hat{\bm{R}}_{m_i}=
    \begin{cases}
     \bm{I} , & m_i = 0\\
    (1-\alpha_I)\Tilde{\bm{A}} + \alpha_I \bm{I} , & m_i > 0
    \end{cases}$
    }
$\hat{\bm{Y}} = (\hat{\bm{R}}_{m_1}\bm{X} \oplus \hat{\bm{R}}_{m_2}\bm{X} \oplus \cdots \oplus \hat{\bm{R}}_{m_s}\bm{X}) \bm{\Theta}_{priv}$\;
}
\Else{
Compute $\bm{Z}$ by Eq. (\ref{eq:concatenate Z}) on $D$\;
$\hat{\bm{Y}} = \bm{Z} \bm{\Theta}_{priv}$\;}\label{line:non-private inference}
\Return $\hat{\bm{Y}}$\;
\end{algorithm}

\subsection{Convexity}\label{app:convexity}
Our proposed solution is based on the general framework of objective perturbation, first introduced in DPERM~\cite{chaudhuri2011DPERM}. The main idea is that to release a trained machine learning model, we randomly perturb its training objective function to satisfy DP; note that the objective function covers the underlying training dataset. Then, model fitting can be performed with an arbitrary optimization algorithm, based on the noisy objective function. This framework requires that the objective function is strongly convex and has bounded derivatives with respect to the model parameters. The concept of strong convexity is defined as follows.

\begin{definition}[Convex and strictly convex]\label{def:convex}
A function $F:\mathbb{R}^{d_1 \times d_2} \rightarrow \mathbb{R}$ is convex if $\forall \bm{\Theta}_1, \bm{\Theta}_2 \in \mathbb{R}^{d_1 \times d_2}$ and $\forall \alpha \in (0,1)$, we have
\begin{equation*}
    F(\alpha \bm{\Theta}_1 + (1-\alpha) \bm{\Theta}_2) \leq \alpha F(\bm{\Theta}_1) + (1-\alpha) F(\bm{\Theta}_2)
\end{equation*}
In particular, $F$ is strictly convex if $\forall \bm{\Theta}_1, \bm{\Theta}_2 \in \mathbb{R}^{d_1 \times d_2}, \bm{\Theta}_1 \neq \bm{\Theta}_2$ and $\forall \alpha \in (0,1)$, we have
\begin{equation*}
    F(\alpha \bm{\Theta}_1 + (1-\alpha) \bm{\Theta}_2) < \alpha F(\bm{\Theta}_1) + (1-\alpha) F(\bm{\Theta}_2)
\end{equation*}
\end{definition}
\begin{definition}[Strongly convex]\label{def:strongly convex}
A function $F(\bm{\Theta})$ over $\bm{\Theta} \in \mathbb{R}^{d_1 \times d_2}$ is strongly convex if $\exists \lambda>0$, $F(\bm{\Theta}) - \lambda \|\bm{\Theta}\|_F^2$ is convex.
\end{definition}
\begin{fact}\label{fact:addition of strongly convex}
     Consider any $\Lambda>0$. $\Lambda\|\bm{\Theta}\|_F^2$ is strongly convex. In addition, if $F(\bm{\Theta})$ is convex, then $F(\bm{\Theta}) + \Lambda\|\bm{\Theta}\|_F^2$ is strongly convex.
\end{fact}

\subsection{Lemma~\ref{lemma:invertible} and its Proof}\label{app:lemma invertible}
In this section we prove matrix $\bm{I}-(1-\alpha) \Tilde{\bm{A}}$ is invertible.
\begin{lemma}\label{lemma:invertible}
    Matrix $\bm{I}-(1-\alpha) \Tilde{\bm{A}}$ is invertible.
\end{lemma}
\begin{proof}
The lemma holds iff the determinant $\det(\bm{I}-(1-\alpha) \Tilde{\bm{A}}) \neq 0$, i.e., $\det(\Tilde{\bm{A}} - \frac{1}{1-\alpha}\bm{I}) \neq 0$, which is the case iff $\frac{1}{1-\alpha}$ is not a eigenvalue of $\Tilde{\bm{A}}$ (by the fundamental theorem of algebra). Here we know that $\frac{1}{1-\alpha}$ is always larger than $1$ since $\alpha \in (0,1)$. Next we prove that any eigenvalue $\lambda$ of $\Tilde{\bm{A}}$ satisfies $|\lambda|\leq 1$.

For any eigenvalue $\lambda$ of $\Tilde{\bm{A}}$, let its corresponding eigenvector be $\bm{x}=(x_1,\cdots,x_N)$, i.e.,
\begin{equation*}
    \Tilde{\bm{A}} \bm{x} = \lambda \bm{x}.
\end{equation*}
Let $i$ be the subscript such that $|x_i|\geq |x_j|, \forall j \in [1,n]$ and let $\Tilde{a}_{ij}$ be the entry on the $j$-th column of the $i$-th row of matrix $\Tilde{\bm{A}}$. Focusing on the $i$-th dimension on both sides of the above equation, we have that 
\begin{equation*}
    \sum_j^n \Tilde{a}_{ij} x_j = \lambda x_i.
\end{equation*}
Taking the absolute values of both sides, we get
\begin{align*}
    |\lambda x_i| &= |\sum_j^n \Tilde{a}_{ij}x_j| \leq \sum_j^n \Tilde{a}_{ij}|x_j| \leq \sum_j^n \Tilde{a}_{ij} |x_i| = |x_i|.
\end{align*}
The first inequality follows from the fact that every element of $\Tilde{\bm{A}}$ is non-negative (either $0$ or $1$) and the last equality follows since $\sum_j^n \Tilde{a}_{ij}$=1  ($\Tilde{\bm{A}}=\bm{D}^{-1}\Tilde{A}$ by definition).  Hence, $|\lambda|\leq 1 < \frac{1}{1-\alpha}$ and the proof follows.
\end{proof}

\subsection{Lemmas Related to \texorpdfstring{\(\mathcal{L}_{\Lambda}\)}{}}\label{app:lemmas_l_lambda}
\begin{lemma}\label{lemma:convexity}
$L(\bm{\Theta};\bm{z}_i,\bm{y}_i)$ is convex w.r.t. $\bm{\Theta}$. $L_{\Lambda}(\bm{\Theta};\bm{Z},\bm{Y})$ is strongly convex w.r.t. $\bm{\Theta}$.
\end{lemma}
\begin{proof}
When $y\in \{0,1\}$, $\ell''(x;y)>0$, so $\ell(x;y)$ is convex. The Hessian matrix of $\ell(\bm{z}_i^T \bm{\theta}_j; \bm{y}_{ij})$ w.r.t. $\bm{\theta}_j$ is $z_iz_i^T \ell''(\bm{z}_i^T \bm{\theta}_j; \bm{y}_{ij})$, whose eigenvalues is non-negative. Hence, $\ell(\bm{z}_i^T \bm{\theta}_j; \bm{y}_{ij})$ is convex w.r.t. $\bm{\theta}_j$.

The Hessian matrix of $L(\bm{\Theta};\bm{z}_i,\bm{y}_i)$ w.r.t. $\bm{\Theta}$ is a block matrix with $z_iz_i^T \ell''(\bm{z}_i^T \bm{\theta}_j; \bm{y}_{ij}), j\in [1,c]$ locating on its diagonal and 0 on elsewhere. Hence, its eigenvalues is non-negative, i.e., it is convex w.r.t. $\bm{\Theta}$. Then $\sum_i^n L(\bm{\Theta};\bm{z}_i,\bm{y}_i)$ is convex w.r.t. $\bm{\Theta}$. Furthermore, since $\frac{1}{2}\Lambda \|\bm{\Theta}\|_F^2$ is strongly convex w.r.t. $\bm{\Theta}$, $L_{\Lambda}(\bm{\Theta};\bm{Z},\bm{Y})$ is strongly convex w.r.t. $\bm{\Theta}$. 
\end{proof}

\begin{lemma}\label{lemma:bijection}
    Given fixed graph dataset $D$, i.e., fixed $\bm{Z}$ and $\bm{Y}$, for any noise $\bm{B}$, the optimal solution $\bm{\Theta}_{priv}$ is unique. Moreover, the mapping from $\bm{B}$ (or $\bm{B}'$ on dataset $D'$) to $\bm{\Theta}_{priv}$ constructed by Eq.~(\ref{eq:B and Theta}) is bijective and continuous differentialable w.r.t. $\bm{\Theta}_{priv}$.
\end{lemma}
\begin{proof}
As stated in Lemma~\ref{lemma:convexity}, $L_{\Lambda}(\bm{\Theta}_{priv};\bm{Z},\bm{Y})$ is strongly convex w.r.t. $\bm{\Theta}$. By Fact~\ref{fact:addition of strongly convex}, $L_{\Lambda}(\bm{\Theta}_{priv};\bm{Z},\bm{Y}) + \frac{1}{2}\Lambda'\|\bm{\Theta}\|_F^2$ is strongly convex. Hence, for a fixed $\bm{B}$, the optimal solution $\bm{\Theta}_{priv}$ is unique and the mapping from $\bm{B}$ to $\bm{\Theta}_{priv}$ is injective.

Eq.~(\ref{eq:B and Theta}) shows that for any $\bm{\Theta}_{priv}$, there exists a $\bm{B}$ for which $\bm{\Theta}_{priv}$ is the minimizer so that the mapping from $\bm{B}$ to $\bm{\Theta}_{priv}$ is surjective.

In conclusion, for any graph dataset $D$, i.e., fixed $\bm{Z}$ and $\bm{Y}$, the mapping between $\bm{\Theta}_{priv}$ and $\bm{B}$ is bijective. As defined in Section~\ref{sec:GCON pert}, $L_{\Lambda}$ has higher order and continuous derivatives. Hence, the mapping is continuous differentiable w.r.t. $\bm{\Theta}_{priv}$. Similarly, the argument holds for $\bm{B}'$ on $D'$.
\end{proof}

\subsection{Correctness of Algorithm~\ref{alg:sampling}}\label{app:Gaussian rotational symmetry}
Here we prove the correctness of our sampling method (Algorithm~\ref{alg:sampling}).
\begin{lemma}\label{lemma:Gaussian rotational symmetry}
Algorithm~\ref{alg:sampling} uniformly samples a vector $\bm{b}$ on the $d$-dimensional hypersphere with radius $a$.
\end{lemma}
\begin{proof}
We have vector $\bm{u}=(u_1,u_2,\cdots,u_d) \sim \mathcal{N}\left(0, \bm{I}_d\right)$. According to the rotational symmetry of the normal distribution, for any orthogonal matrix $\bm{S}$, we can rotate $\bm{u}$ to $\bm{S}\bm{u}$, which is oriented in another direction and satisfies $\bm{S}\bm{u} \sim \mathcal{N}\left(0, \bm{I}_d\right)$, i.e., $\bm{u}$ and $\bm{S}\bm{u}$ follow the same distribution. Let $\overline{\bm{u}}= a\frac{\bm{u}}{\|\bm{u}\|}$ and $\overline{\bm{u_S}} = a\frac{\bm{S}\bm{u}}{\|\bm{S}\bm{u}\|} = a\frac{\bm{S}\bm{u}}{\|\bm{u}\|}$. $\overline{\bm{u}}$ and $\overline{\bm{u_S}}$ have the same length $a$ and follow the same distribution for any rotation matrix $\bm{S}$, i.e., each direction on the hypersphere with radius $a$ is sampled with equal probability. Hence, the lemma is validated.
\end{proof}

\subsection{Loss Functions}\label{app:derivatives of l(x;y)}
The following shows the specific expressions of $\ell(x; y)$ in the MultiLabel Soft Margin loss and the Pseudo-Huber loss. Additionally, we list their first to third derivatives, denoted as $\ell'(x;y)$, $\ell''(x;y)$, and $\ell'''(x;y)$, respectively.

\noindent\textbf{MultiLabel Soft Margin loss}
\begin{equation}\label{eq:mlsm loss}
\begin{aligned}
    \ell(x; y) = -\frac{1}{c} \left(y \log(\frac{1}{1+e^{-x}}) + (1-y)\log(\frac{e^{-x}}{1+e^{-x}}) \right) ,\\
\end{aligned}
\end{equation}

\begin{equation*}
\begin{aligned}
    \ell'(x;y) &= -\frac{1}{c} \left(y\frac{1}{1+e^x} + (1-y) \frac{-e^x}{1+e^x} \right)\\
    \ell''(x;y) &= -\frac{1}{c} \left(y\frac{-e^x}{(1+e^x)^2} + (1-y) \frac{-e^x}{(1+e^x)^2} \right)\\
    \ell'''(x;y) &= -\frac{1}{c} \left(y\frac{e^x(e^x-1)}{(1+e^x)^3} + (1-y) \frac{e^x(e^x-1)}{(1+e^x)^3} \right)\\
\end{aligned}
\end{equation*}
The supremums of the absolute values of them are
\begin{equation*}
    \sup(|\ell'|) = \frac{1}{c},\quad \sup(|\ell'|)=\frac{1}{4c},\quad \sup(|\ell'''|) = \frac{1}{6\sqrt{3}c}
\end{equation*}

\noindent\textbf{Pseudo-Huber loss}
\begin{equation}\label{eq:ph loss}
\begin{aligned}
    &\ell(x; y) = \frac{\delta_l^2}{c} \left( \sqrt{1+\frac{(x-y)^2}{\delta_l^2}}-1 \right) ,\\
\end{aligned}
\end{equation}
where the weight $\delta_l$ is a hyperparameter.
\begin{equation*}
\begin{aligned}
    \ell'(x;y) &= \frac{x-y}{c \sqrt{\frac{(x-y)^2}{\delta_l^2}+1}} \\
    \ell''(x;y) &= \frac{1}{c \cdot\left(\frac{(x-y)^2}{\delta_l^2}+1\right)^{\frac{3}{2}}} \\
    \ell'''(x;y) &= -\frac{3(x-y)}{c \delta_l^2 \cdot\left(\frac{(x-y)^2}{\delta_l^2}+1\right)^{\frac{5}{2}}}\\
\end{aligned}
\end{equation*}
The supremums of the absolute values of them are
\begin{equation*}
    \sup(|\ell'|) = \frac{\delta_l}{c},\quad \sup(|\ell'|)=\frac{1}{c},\quad \sup(|\ell'''|) = \frac{48\sqrt{5}}{125 c \delta_l}
\end{equation*}

\subsection{Proof of Lemma~\ref{lemma:Psi_Z of PPR and APPR}}\label{app:Proof of Psi_Z}
\begin{proof}
We prove Eqs. (\ref{eq:Psi Z_m}) \& (\ref{eq:Psi Z}) sequentially.

\noindent\textbf{Proof for Eq. (\ref{eq:Psi Z_m}).}
We prove Eq. (\ref{eq:Psi Z_m}) by examining the distinctions of $\Tilde{\bm{A}}$, $\bm{R}_m$, and $\bm{Z}_m$ on neighboring graphs one by one. Note that the computation of $\bm{R}_m$ and $\bm{Z}_m$ differs based on two propagation schemes: PPR ($m=\infty$) and APPR ($m\geq 0$). We begin with the case of PPR ($m=\infty$).

Without loss of generality, we assume that $D'$ is obtained from removing one edge between node $1$ and node $2$ in $D$. Let $k_i$ be the degree of node $i$ in $D$. We compute $\Tilde{\bm{A}}$, $\bm{R}_{\infty}$, and $\bm{Z}_{\infty}$ on $D$ and $\Tilde{\bm{A}}'$, $\bm{R}_{\infty}'$, and $\bm{Z}'_{\infty}$ on  $D'$.
according to Section~\ref{sec:GCON Propagation}.
When removing the edge, by definition, only the first and second rows in $\Tilde{\bm{A}}$ change. In particular, both $\Tilde{\bm{A}}_{12}$ and $\Tilde{\bm{A}}_{21}$ change from $\frac{1}{k_1+1}$ and $\frac{1}{k_2+1}$ to 0. $\Tilde{\bm{A}}_{11}$ and the entries of $k_1-1$ neighbors of node $1$ change from $\frac{1}{k_1+1}$ to $\frac{1}{k_1}$ respectively. It is similar for $\Tilde{\bm{A}}_{22}$ and the $k_2-1$ neighbors of node $2$.
Hence, we have
\begin{equation}\label{def:delta A}
    \Delta\Tilde{\bm{A}} \triangleq \Tilde{\bm{A}}' - \Tilde{\bm{A}} = \bm{v}_1 \bm{a}_1^T + \bm{v}_2 \bm{a}_2^T.
\end{equation}
Here $\bm{v}_1 = (1, 0, \dots, 0)^{T}$ and $\bm{v}_2 = (0, 1, 0, \dots, 0)^{T}$. $\bm{a}_1 \in \mathbb{R}^{n}$ is a vector whose second entry is $-\frac{1}{k_1+1}$, $k_1$ entries are $\frac{1}{k_1(k_1+1)}$, and others are 0. $\bm{a}_2$ is similar. Then,
\begin{equation}\label{eq:delta R definition}
\begin{aligned}
    \Delta\bm{R}_{\infty} \triangleq \bm{R}_{\infty}' - \bm{R}_{\infty} & = \bm{R}_{\infty}' \bm{R}_{\infty}^{-1} \bm{R}_{\infty} - \bm{R}_{\infty}' \bm{R}_{\infty}'^{-1} \bm{R}_{\infty} \\
    & = \frac{1- \alpha}{\alpha} \bm{R}_{\infty}' \Delta\Tilde{\bm{A}} \bm{R}_{\infty} ,
\end{aligned}
\end{equation}

\begin{equation}\label{eq:delta_Z definition}
\begin{aligned}
  \Delta\bm{Z}_{\infty} \triangleq \bm{Z}_{\infty}' - \bm{Z}_{\infty} & = \Delta\bm{R}_{\infty} \bm{X} \\
  & = \frac{1- \alpha}{\alpha} \bm{R}_{\infty}' \Delta\Tilde{\bm{A}} \bm{R}_{\infty} \bm{X}\\
  & = \frac{1- \alpha}{\alpha} \bm{R}_{\infty}'(\bm{v}_1 \bm{a}_1^T \bm{Z}_{\infty} + \bm{v}_2 \bm{a}_2^T \bm{Z}_{\infty}) .
\end{aligned}
\end{equation}
By Lemma~\ref{lemma:AR properties} with $p=\frac{1}{2}$, the sum of column $i$ of $\bm{R}_{\infty}'$ is at most $\max(\frac{k_i'+1}{2}, 1)$, where $k_i'$ is the degree of node $i$ in $D'$. Since $D'$ is obtained from $D$ by removing one edge between nodes 1 and 2, we have $k_1\geq 1$ and $k_2\geq 1$. Hence, we have that
\begin{equation}~\label{eq:bound R'v}
\begin{aligned}
    \|\bm{R}_{\infty}' \bm{v}_1\|_1 \leq \frac{k_1+1}{2},\quad \|\bm{R}_{\infty}' \bm{v}_2\|_1 \leq \frac{k_2+1}{2}.
\end{aligned} 
\end{equation}
Let $\bm{x}_i$, $\bm{r}_{\infty,i}$ and $\bm{z}_{\infty,i}$ denote the $i$-th rows of $\bm{X}$, $\bm{R}_{\infty}$ and $\bm{Z}_{\infty}$, respectively. Each feature $\|\bm{x}_i\|_2$ is normalized to 1. By Lemma~\ref{lemma:AR properties}, the sum of each row of $\bm{R}_{\infty}$ is 1, for all $i \in [1,n]$, we have
\begin{equation*}
    \|\bm{z}_{\infty,i}\|_2 = \|\bm{r}_{\infty,i}^T\bm{X}\|_2 = \|\bm{r}_{\infty,i}^T\|_1 \max_i\|\bm{x}_i\|_2 \leq 1 .
\end{equation*}
Then we have that
\begin{equation}\label{eq:bound aZ}
    \|\bm{a}_1^T \bm{Z}_{\infty}\|_2 
    \leq \frac{2}{k_1+1},\quad \|\bm{a}_2^T \bm{Z}_{\infty}\|_2\leq \frac{2}{k_2+1}
\end{equation}
Using Eqs.~(\ref{eq:bound R'v}) \& (\ref{eq:bound aZ}),
\begin{equation}\label{eq:sum delta z of PPR}
\begin{aligned}
    \psi(\bm{Z}_{\infty}) &\leq \|\bm{R}_{\infty}' \bm{v}_1\|_1 \|\bm{a}_1^T \bm{Z}_{\infty}\|_2 +\|\bm{R}_{\infty}' \bm{v}_2\|_1 \|\bm{a}_2^T \bm{Z}_{\infty}\|_2 \\ 
    &\leq \frac{2 (1 - \alpha)}{\alpha} \triangleq \Psi(\bm{Z}_{\infty}).
\end{aligned}
\end{equation}
Eq. (\ref{eq:sum delta z of PPR}) quantifies the changes in the $\mathcal{L}_2$ norms of rows in $\bm{Z}_{\infty}$ when modifying one edge in the PPR scheme.

Next, we consider the other case of APPR ($m\geq 0$). For $m=0$, we have 
\begin{equation*}
    \psi(\bm{Z}_0) = \psi(\bm{I} \bm{X}) = 0 = \frac{2(1-\alpha)}{\alpha}[1-(1-\alpha)^0].
\end{equation*}
For any $m\geq 1$, we define
\begin{equation*}
    \triangle \Tilde{\bm{A}}^m \triangleq \Tilde{\bm{A}}'^m - \Tilde{\bm{A}}^m.
\end{equation*}
We have
\begin{equation}\label{eq:R_m to A_m}
    \bm{R}_m' - \bm{R}_m = \alpha\sum_{i=1}^{m-1} (1-\alpha)^i\triangle \Tilde{\bm{A}}^i + (1-\alpha)^m \triangle \Tilde{\bm{A}}^m .
\end{equation}
By triangle inequality,
\begin{equation}\label{eq:psi_Z of R_m to psi_Z A^m}
\begin{aligned}
    \psi(\bm{Z}_m) \leq& \alpha\sum_{i=1}^{m-1} (1-\alpha)^i\psi(\Tilde{\bm{A}}^i \bm{X}) + (1-\alpha)^m \psi(\Tilde{\bm{A}}^m \bm{X}) .
\end{aligned}
\end{equation}
Hence, the key to calculate the bounds for $\psi(\bm{Z}_m)$ is calculating $\psi(\Tilde{\bm{A}}^m \bm{X})$. We have known that $\triangle \Tilde{\bm{A}} = \bm{v}_1\bm{a}_1^T + \bm{v}_2\bm{a}_2^T$. For any two matrices $\bm{M}_1$ and $\bm{M}_2$, the matrix $\bm{M}_1 \bm{v}_1\bm{a}_1^T \bm{M}_2$ is still in the form of a column by a row, of which the sum of norm of rows can be bounded similarly as Eqs.~(\ref{eq:bound R'v}) \& (\ref{eq:bound aZ}). Therefore, the key to calculating $\psi(\Tilde{\bm{A}}^m \bm{X})$ is splitting $\triangle \Tilde{\bm{A}}^m$ into addition of $\bm{M}_1\triangle\Tilde{\bm{A}}\bm{M}_2$, which is an addition of a column by a row.
When $m=1$,
\begin{equation*}
    \psi(\Tilde{\bm{A}} \bm{X}) \leq 2\frac{2}{k_1+1} \leq 2 .
\end{equation*}
When $m\geq 2$,
\begin{equation}\label{eq:recursively decompose A^m}
\begin{aligned}
    \triangle\Tilde{\bm{A}}^m = \frac{1}{2}[&\triangle\Tilde{\bm{A}}^{m-1}(\Tilde{\bm{A}}' + \Tilde{\bm{A}}) + (\Tilde{\bm{A}}'^{m-1} + \Tilde{\bm{A}}^{m-1})\triangle\Tilde{\bm{A}}] .\\
\end{aligned}
\end{equation}
As shown above, we can recursively decompose $\triangle\Tilde{\bm{A}}^m$ to $\triangle\Tilde{\bm{A}}^{m-1}$, $\triangle\Tilde{\bm{A}}^{m-2}$, $\cdots$ and finally an addition of $\bm{M}_1\triangle\Tilde{\bm{A}}\bm{M}_2$. Then following a similar derivation as Eqs.~(\ref{eq:bound R'v}) \& (\ref{eq:bound aZ}), we have
\begin{equation}\label{eq:psi A^m}
    \psi(\Tilde{\bm{A}}^m \bm{X}) \leq 2m. 
\end{equation}
Combining Eqs.~(\ref{eq:psi A^m}) \& (\ref{eq:psi_Z of R_m to psi_Z A^m}),
\begin{equation}\label{eq:psi_Z of R_m}
\begin{aligned}
    \psi(\bm{Z}_m) &\leq \frac{2(1-\alpha)}{\alpha}[1-(1-\alpha)^m] = \Psi(\bm{Z}_m) .
\end{aligned}
\end{equation}
Hence, Eq.~(\ref{eq:Psi Z_m}) holds for any $m\in [0,\infty]$.

\noindent\textbf{Proof for Eq.~(\ref{eq:Psi Z}).}
Let $\bm{z}_i$ and $\bm{z}_{m,i}$ denote the $i$-th rows of $\bm{Z}$ and $\bm{Z}_{m}$, respectively. We have
\begin{align*}
    \psi(\bm{Z}) = \sum_{i=1}^{n}\|\bm{z}_i' - \bm{z}_i\|_2 
    &\leq \sum_{i=1}^{n} \frac{1}{s} \sum_{j=1}^s \|\bm{z}_{m_j, i}' - \bm{z}_{m_j, i}\|_2 \\
    &= \frac{1}{s} \sum_{j=1}^s \psi(\bm{Z}_{m_j}) \leq \frac{1}{s} \sum_{j=1}^s \Psi(\bm{Z}_{m_j}),
\end{align*}
where the first inequality follows from the triangle inequality. Hence, Eq.~(\ref{eq:Psi Z}) holds.

In summary, Lemma~\ref{lemma:Psi_Z of PPR and APPR} holds.
\end{proof}

\subsection{Proof of Lemma~\ref{lemma:AR properties}}\label{app:AR properties}
\begin{proof}

For the first conclusion, following our normalization, every entry of $\Tilde{\bm{A}}$ is non-negative. Since $\Tilde{\bm{A}}^m$, $\bm{R}_m$ and $\bm{R}_{\infty}$ are polynomials w.r.t. $\Tilde{\bm{A}}$ with positive coefficients, every entry of them is non-negative. 

For the second conclusion, by definition, the sum of each row of $\Tilde{\bm{A}}$ is 1. Then the sum of each row of $\Tilde{\bm{A}}^m$ is 1. Using the fact that the sum of each row of $\bm{I}$ is 1, for each row of $\bm{R}_m$, its sum is
\begin{align*}
    1 \cdot \alpha \sum_{j=0}^{m-1} (1 - \alpha)^j + 1 \cdot (1 - \alpha)^m = 1
\end{align*}
The above equation holds for any $m>0$ including $m\rightarrow \infty$, i.e., $\bm{R}_m \rightarrow \bm{R}_{\infty}$.


Next, we prove the third conclusion.
Since every entry of $\Tilde{\bm{A}}$ is non-negative, we get that every entry of $\Tilde{\bm{A}}^m$ is non-negative. Let $\bm{v}_{m,i}$ be the $i$th column of $\Tilde{\bm{A}}^m$. Let $N_i$ be the set of neighbors of $i$ and let $k_i = |N_i|$. We show by induction on $m$ for any $m \geq 1$ that: for any $i$, $|\bm{v}_{m,i}|_1 \leq \max((k_i+1)p, 1)$. In our proof, we leverage that $\bm{A}$ is symmetric.

When $m = 1$, $\Tilde{\bm{A}}^{m} = \Tilde{\bm{A}}$.
\begin{align*}
    |\bm{v}_{m,i}|_1 &= \sum_{j \in N_i} \min(\frac{1}{k_j+1},p) + (1- k_i\min(\frac{1}{k_i+1},p))\\
    &\leq k_i p + \max(\frac{1}{k_i+1},1-k_i p)\\
    &=\max(\frac{1}{k_i+1} + k_i p,1)\\
\end{align*}
\begin{enumerate}
    \item If $\frac{1}{k_i+1} < p$, $\max(\frac{1}{k_i+1} + k_i p,1)\leq (k_i+1)p = \max((k_i+1)p, 1)$.
    \item If $\frac{1}{k_i+1} \geq p$, $\max(\frac{1}{k_i+1} + k_i p,1)\leq 1 = \max((k_i+1)p, 1)$.
\end{enumerate}
Hence, the lemma holds when $m=1$.

When $m > 1$ and $|\bm{v}_{m,i}|_1 \leq \max((k_i+1)p, 1)$.
As $\Tilde{\bm{A}}^{m+1} = \Tilde{\bm{A}}^m \cdot \Tilde{\bm{A}}$, we have
\begin{align*}
    |\bm{v}_{m+1,i}|_1 =& \sum_{j \in N_i} \min(\frac{1}{k_j+1},p) |\bm{v}_{m,j}|_1 + \\
    &(1- k_i\min(\frac{1}{k_i+1},p)) |\bm{v}_{m,i}|_1 \\
    \leq& p\sum_{j \in N_i} \min(\frac{1}{(k_j+1)p},1) \max((k_j+1)p, 1) + \\
    &(1- k_i\min(\frac{1}{k_i+1},p)) \max((k_i+1)p, 1)\\
    =& k_ip + \max(\frac{1}{k_i+1},1-k_i p)\max((k_i+1)p, 1).
\end{align*}
\begin{enumerate}
    \item If $\frac{1}{k_i+1} < p$,
    \begin{align*}
        & k_ip + \max(\frac{1}{k_i+1},1-k_i p)\max((k_i+1)p, 1) \\
        \leq& k_ip + \frac{1}{k_i+1} \cdot (k_i+1)p \\
        =& (k_i+1)p = \max((k_i+1)p, 1).
    \end{align*}
    \item If $\frac{1}{k_i+1} \geq p$,
    \begin{align*}
        & k_ip + \max(\frac{1}{k_i+1},1-k_i p)\max((k_i+1)p, 1) \\
        \leq& k_ip + (1-k_i p) \cdot 1 \\
        =& 1 = \max((k_i+1)p, 1).
    \end{align*}
\end{enumerate}
Hence, $|\bm{v}_{m+1,i}|_1 \leq \max((k_i+1)p, 1)$, i.e., the lemma holds for any $m \geq 1$.

For $\bm{R}_m$, by Eq.~(\ref{eq:APPR matrix}),
\begin{align*}
    \bm{R}_m 
    =& \alpha \sum_{i=0}^{m-1} (1 - \alpha)^i  \Tilde{\bm{A}}^i + (1-\alpha)^m \Tilde{\bm{A}}^m.
\end{align*}

Using the above conclusion of $\Tilde{\bm{A}}^m$, the sum of column $i$ of $\bm{R}_m$ satisfies
\begin{align*}
    \leq& \left(\alpha \sum_{j=0}^{m-1} (1 - \alpha)^j + (1 - \alpha)^m \right) \cdot \max((k_i+1)p, 1) =\max((k_i+1)p, 1).
\end{align*}
The above inequality holds for $\bm{R}_m$ with any $m$, so it holds for $m\rightarrow \infty$, i.e., $\bm{R}_m \rightarrow \bm{R}_{\infty}$. 

In summary, we have the sum of column $i$ of $\Tilde{\bm{A}}^m, \bm{R}_m$, or $\bm{R}_{\infty}$ is $\leq \max((k_i+1)p, 1)$.
\end{proof}

\subsection{Proof of Theorem~\ref{theorem:train DP}}\label{app:proof of Theorem}
\begin{proof}
First, we bound the ratio $\frac{g(\mathbf{\Theta}_{priv}|D)}{g(\mathbf{\Theta}_{priv}|D')}$ of the densities of $\bm{\Theta}_{priv}$ (from Eq. (\ref{eq:Theta_priv})) for neighboring graphs $D$ and $D'$. 

Let $\bm{B}=(\bm{b}_1, \bm{b}_2,\cdots,\bm{b}_c)$ and $\bm{B}'=(\bm{b}_1', \bm{b}_2', \cdots,\bm{b}_c')$ denote the noise matrices sampled by the training algorithm on $D$ and $D'$ respectively. Let $\bm{\Theta}_{priv}$ be the optimal solution on both $L_{priv}(\bm{\Theta};\bm{Z},\bm{Y})$ and $L_{priv}(\bm{\Theta};\bm{Z}',\bm{Y})$, where $\bm{Z}$ and $\bm{Z}'$ are computed on $D$ and $D'$, respectively. As $\bm{\Theta}_{priv}$ is optimal, the derivatives of the two loss functions are both 0 at $\bm{\Theta} = \bm{\Theta}_{priv}$; accordingly, we have
\begin{equation}\label{eq:B and Theta}
\begin{aligned}
    \bm{B} = -n_1\frac{\partial \left(L_{\Lambda}(\bm{\Theta}_{priv};\bm{Z},\bm{Y}) + \frac{1}{2}\Lambda' \|\bm{\Theta}_{priv}\|_F^2 \right)}{\partial \bm{\Theta}_{priv}} ,\\
    \bm{B}' = -n_1\frac{\partial \left(L_{\Lambda}(\bm{\Theta}_{priv};\bm{Z}',\bm{Y}) + \frac{1}{2}\Lambda' \|\bm{\Theta}_{priv}\|_F^2 \right)}{\partial \bm{\Theta}_{priv}}.
\end{aligned}
\end{equation}

By Lemma~\ref{lemma:bijection} that the mappings from $\bm{\Theta}_{priv}$ to $\bm{B}$ and $\bm{B}'$ in Eq.~(\ref{eq:B and Theta}) are continuously differentiable and bijective, we can use Eq.~(20.20) in \cite{billingsley2017probability} that the PDF of $\bm{\Theta}_{priv}$ on graphs $D$ and $D'$ can be expressed as
\begin{equation}\label{eq:Theta density split}
    \frac{g(\bm{\Theta}_{priv} | D)}{g(\bm{\Theta}_{priv} | D')}
  = \frac {|\det(\bm{J}(\bm{\Theta}_{priv} \mapsto \bm{B} | D) )|^{-1}} {|\det(\bm{J}(\bm{\Theta}_{priv} \mapsto \bm{B}' | D') )|^{-1}} \cdot \frac{\mu(\bm{B} | D) }{\mu(\bm{B}' | D')},
\end{equation}
where $\bm{J}(\bm{\Theta}_{priv} \mapsto \bm{B} | D)$ and $\bm{J}(\bm{\Theta}_{priv} \mapsto \bm{B}' | D')$ are the Jacobian matrices of the mappings from $\bm{\Theta}_{priv}$ to $\bm{B}$ and $\bm{B}'$ respectively, and $\mu(\cdot)$ are the probability density functions. Let $\bm{\Theta}_{priv} = (\bm{\theta}_1,\bm{\theta}_2,\cdots, \bm{\theta}_c)$ where $\bm{\theta}_j\in \mathbf{R}^{d}, j \in [1,c]$. For $\bm{\theta}_j, j\in[1, c]$, there are two cases to consider: (i) for all $j$, $\|\bm{\theta}_j\|_2 \leq c_{\theta}$ and (ii) there exists some $j$ that $\|\bm{\theta}_j\|_2 > c_{\theta}$. We consider the case (i) first.

\noindent\textbf{Case (i).} By Lemma~\ref{lemma:bound for Jacobian ratio}, 
\begin{equation}\label{eq:Jacobian ratio of Theta to B}
\begin{aligned}
\frac {|\det(\bm{J}(\bm{\Theta}_{priv} \mapsto \bm{B} | D) )|^{-1} } {|\det(\bm{J}(\bm{\Theta}_{priv} \mapsto \bm{B}' | D') )|^{-1}} 
\leq & \left(1 + \frac{(2c_2 + c_3 c_{\theta}) \Psi(\bm{Z})}{d n_1(\Lambda + \Lambda')} \right)^{cd} ,
\end{aligned}
\end{equation}
where $c_2$, $c_3$, $\Psi(\bm{Z})$, $c_{\theta}$, and $\Lambda'$ follow Eq. (\ref{eq:supremum of l(x;y)}), (\ref{eq:Psi_Z in Theorem}), (\ref{eq:c_theta}), \& (\ref{eq:Lambda_prime}), respectively.
As $\Lambda' \geq 0$, the RHS of Eq.~(\ref{eq:Jacobian ratio of Theta to B}) is bounded by $\exp{(\epsilon_{\Lambda})}$ (as in Eq.~(\ref{eq:epsilon_Lambda})). Then the remaining budget for bounding $\frac{\mu(\bm{B} | D)}{\mu(\bm{B}' | D')}$ is $\exp{(\epsilon - \epsilon_{\Lambda})}$. We compress the above ratio to reserve more budget for the latter. Thus, we set a hyperparameter $\omega$ to artificially bound $\exp{(\epsilon_{\Lambda})}$ by $\exp{(\epsilon - \omega \epsilon)}$. Then there are two cases to consider: (I) $\epsilon_{\Lambda} \leq (1-\omega)\epsilon$, which means that even if $\Lambda'=0$, the remaining budget ($\geq \omega \epsilon$) is sufficient, and (ii) $\epsilon_{\Lambda} > (1-\omega)\epsilon$, for which we set the additional parameter $\Lambda'$ by Eq. (\ref{eq:Lambda_prime}). This ensures that the ratio can be bounded by $\exp{(\epsilon - \omega \epsilon)}$, with an exact budget of $\omega \epsilon$ remaining. Therefore,
\begin{equation}\label{eq:bound for Jacobian ratio}
    \frac {|\det(\bm{J}(\bm{\Theta}_{priv} \mapsto \bm{B} | D) )|^{-1} } {|\det(\bm{J}(\bm{\Theta}_{priv} \mapsto \bm{B}' | D') )|^{-1}} 
    \leq \exp{(\min(\epsilon_{\Lambda}, \epsilon-\omega\epsilon))} .
\end{equation}

By Lemma~\ref{lemma:bound for mu ratio} and Eq.~(\ref{eq:beta}), we have
\begin{equation}\label{eq:bound for mu B}
\begin{aligned}
    \frac{\mu(\bm{B} | D)}{\mu(\bm{B}' | D')}
    &\leq \exp{(\max(\epsilon-\epsilon_{\Lambda}, \omega\epsilon))} .
\end{aligned}
\end{equation}

Combining Eq.~(\ref{eq:Theta density split}), (\ref{eq:bound for Jacobian ratio}), \& (\ref{eq:bound for mu B}), when $\|\bm{\theta}_j\|_2 \leq c_{\theta}$ for all $j\in [1,c]$, we have:
\begin{equation}\label{ineq:good case bound}
\begin{split}
    \frac{g(\bm{\Theta}_{priv} | D)}{g(\bm{\Theta}_{priv} | D')}
  &\leq \exp{(\min(\epsilon_{\Lambda}, \epsilon-\omega\epsilon))} \exp{(\max(\epsilon-\epsilon_{\Lambda}, \omega\epsilon))} \\
  &= \exp{\left(\epsilon \right)}. 
\end{split}
\end{equation}
Similarly, we have $\frac{g(\bm{\Theta}_{priv} | D')}{g(\bm{\Theta}_{priv} | D)}\leq \exp{(\epsilon)}$ when $\|\bm{\theta}_j\|_2 \leq c_{\theta}$ for all $j\in [1,c]$.

\noindent\textbf{Case (ii).} For the case (ii) that there exists some $j$ that $\|\bm{\theta}_{j}\|_2 > c_{\theta}$, we bound the probability that it happens by $\delta$ by Lemma~\ref{lemma:bad case}.

By the analysis results of the two cases, Algorithm~\ref{alg:Train} satisfies ($\epsilon,\delta$)-pDP (defined in Definition~\ref{def:pDP} in Appendix~\ref{app:pDP}). By Lemma~\ref{lemma:pDP to DP}, Algorithm~\ref{alg:Train} satisfies ($\epsilon,\delta$)-DP.
\end{proof}

\subsection{Lemma~\ref{lemma:bound for Jacobian ratio}}\label{app:proof of Jacobian ratio}
\begin{lemma}\label{lemma:bound for Jacobian ratio}
Following the setting in Theorem~\ref{theorem:train DP}, when $\|\bm{\theta}_j\| \leq c_{\theta}$ for all $j\in [1,c]$,
\begin{equation*}
\begin{aligned}
\frac {|\det(\bm{J}(\bm{\Theta}_{priv} \mapsto \bm{B} | D) )|^{-1} } {|\det(\bm{J}(\bm{\Theta}_{priv} \mapsto \bm{B}' | D') )|^{-1}} 
\leq & \left(1 + \frac{(2c_2 + c_3 c_{\theta}) \Psi(\bm{Z})}{d n_1(\Lambda + \Lambda')} \right)^{cd} ,
\end{aligned}
\end{equation*}
where $c_2$, $c_3$, $\Psi(\bm{Z})$, $c_{\theta}$, and $\Lambda'$ follow Eq. (\ref{eq:supremum of l(x;y)}), (\ref{eq:Psi_Z in Theorem}), (\ref{eq:c_theta}), \& (\ref{eq:Lambda_prime}), respectively.
\end{lemma}

\begin{proof}
Observe from Eq.~(\ref{eq:loss function component}) that the derivatives of $\ell(\bm{\Theta};\bm{z}_i,\bm{y}_i)$ with respect to any of $\bm{\theta}_1, \bm{\theta}_2,\cdots, \bm{\theta}_c$ are independent of each other. Therefore, $\bm{J}(\bm{\Theta}_{priv} \mapsto \bm{B} | D)$ and $\bm{J}(\bm{\Theta}_{priv} \mapsto \bm{B}' | D')$ are size $dc\times dc$ block diagonal matrices with size $d\times d$ Jacobian matrices $\bm{J}(\bm{\theta}_j \mapsto \bm{b}_j | D)$ and $\bm{J}(\bm{\theta}_j \mapsto \bm{b}_j' | D')$, $j \in [1,c]$, lying on their diagonals, respectively. Hence,
\begin{equation}\label{eq:Theta to product of theta}
\begin{aligned}
    \frac {|\det(\bm{J}(\bm{\Theta}_{priv} \mapsto \bm{B} | D) )|^{-1} } {|\det(\bm{J}(\bm{\Theta}_{priv} \mapsto \bm{B}' | D') )|^{-1}} = & \prod_{j=1}^{c} \frac{|\det(\bm{J}(\bm{\theta}_j \mapsto \bm{b}_j | D)|^{-1} }{|\det(\bm{J}(\bm{\theta}_j \mapsto \bm{b}_j' | D')|^{-1} }
\end{aligned}
\end{equation}

Let $\ell'(x;y)$, $\ell''(x;y)$, and $\ell'''(x;y)$ be the first, second, and third-order derivatives of $\ell(x;y)$. The specific forms of these derivatives are shown in Appendix~\ref{app:derivatives of l(x;y)}. According to Eq.~(\ref{eq:B and Theta}), $\bm{\theta}_j \mapsto \bm{b}_j | D$ and $\bm{\theta}_j \mapsto \bm{b}_j' | D'$ are constructed by the following equations.
\begin{equation}\label{eq:bj}
\begin{aligned}
    \bm{b}_j & = - \sum_{i=1}^{n_1} \bm{z}_i \ell'(\bm{z}_i^T \bm{\theta}_j;\bm{y}_{ij}) - n_1 (\Lambda+\Lambda') \bm{\theta}_{j} ,\\
  \bm{b}'_j & = - \sum_{i=1}^{n_1} \bm{z}_i' \ell'(\bm{z}_i'^T \bm{\theta}_j;\bm{y}_{ij}) - n_1 (\Lambda+\Lambda') \bm{\theta}_{j} ,
\end{aligned}
\end{equation}

Without loss of generality, we first consider $\bm{\theta}_1$. Define
\begin{align}
  \bm{B}_1 & \triangleq \sum_{i=1}^{n_1} \bm{z}_i \bm{z}_i^T \ell''(\bm{z}_i^T\bm{\theta}_1;\bm{y}_{i1}) + n_1 (\Lambda+\Lambda') \bm{I}_d ,\label{eq:B_1}\\
  \bm{E}_1 & \triangleq - \sum_{i=1}^{n_1}\bm{z}_i \bm{z}_i^T \ell''(\bm{z}_i^T\bm{\theta}_1;\bm{y}_{i1}) + \sum_{i=1}^{n_1} \bm{z}_i' \bm{z}_i'^T \ell''(\bm{z}_i'^T\bm{\theta}_1;\bm{y}_{i1}) ,\label{eq:E_1}
\end{align}
where $\bm{I}_d \in \mathbb{R}^{d\times d}$ is an identity matrix.
Then $\bm{J}(\bm{\theta}_1 \mapsto \bm{b}_1 | D) = - \bm{B}_1 $ and $\bm{J}(\bm{\theta}_1 \mapsto \bm{b}_1' | D') = - (\bm{B}_1 + \bm{E}_1)$. Let $\sigma_i(\cdot)$ denote the $i$-th singular value of matrix $\cdot$, ordered non-increasingly, i.e., $\sigma_1\geq \sigma_2 \geq \cdots$. We have
\begin{align}
\frac{|\det(\bm{J}(\bm{\theta}_1 \mapsto \bm{b}_1 | D)|^{-1} }{|\det(\bm{J}(\bm{\theta}_1 \mapsto \bm{b}_1' | D')|^{-1} }
&= \left|\frac{\det(\bm{B}_1 + \bm{E}_1)}{\det(\bm{B}_1)} \right| \nonumber \\
&= |\det(\bm{B}_1^{-1})\det(\bm{B}_1+\bm{E}_1)| \nonumber \\
&= |\det(I+\bm{B}_1^{-1} \bm{E}_1)| \nonumber \\
&= \prod_i^d \sigma_i(I+\bm{B}_1^{-1} \bm{E}_1) \nonumber \\
&\leq \prod_i^d \left(1+ \sigma_i(\bm{B}_1^{-1} \bm{E}_1) \right) \nonumber \\
\leq& \left(1+\frac{1}{d}\sum_i^{d} \sigma_i(\mathbf{B}_1^{-1} \mathbf{E}_1) \right)^d \label{ineq:GM-AM}\\
\leq& \left(1+\frac{1}{d}\sum_i^{d} \sigma_1(\mathbf{B}_1^{-1}) \sigma_i(\mathbf{E}_1) \right)^d \label{ineq:singular values of product}
\end{align}
where the inequality~\ref{ineq:GM-AM} uses the GM-AM inequality (geometric mean $\leq$ arithmetic mean), the inequality~\ref{ineq:singular values of product} uses Lemma~\ref{lemma:singular values of product}. Since $\bm{B}_1, \bm{E}_1$, and $\bm{B}_1^{-1}$ are Hermite matrices, their ordered singular values equal their absolute eigenvalues, i.e., $\sigma_i=|\lambda_i|, \forall i \in [1,d]$. Observe the definition equation of $\bm{B}_1$ that $\lambda_i(\bm{B}_1) \geq n_1(\Lambda + \Lambda'), \forall i \in [1,d]$. Equivalently, $\forall i \in [1,d]$, 
\begin{equation}\label{eq:singular values of B}
    \sigma_i(\bm{B}^{-1}_1)=\lambda_i(\bm{B}^{-1}_1)\leq \frac{1}{n_1(\Lambda + \Lambda')}.
\end{equation}
Next we analyze $\sigma_i(\mathbf{E}_1)$. We split $\bm{E}_1 = \bm{E}_1^{(1)} + \bm{E}_1^{(2)}$ into two terms where
\begin{equation}\label{eq:E_1(1)E_1(2)}
\begin{aligned}
    \bm{E}_1^{(1)} & \triangleq \sum_{i=1}^{n_1} (\bm{z}_i' \bm{z}_i'^T - \bm{z}_i \bm{z}_i^T) \ell''(\bm{z}_i^T \bm{\theta}_1;\bm{y}_{i1}) ,\\
    \bm{E}_1^{(2)} & \triangleq \sum_{i=1}^{n_1} \bm{z}_i' \bm{z}_i'^T[ \ell''(\bm{z}_i'^T \bm{\theta}_1;\bm{y}_{i1}) - \ell''(\bm{z}_i^T \bm{\theta}_1;\bm{y}_{i1})] .
\end{aligned}
\end{equation}
For $\bm{E}_1^{(1)}$, we have
\begin{equation*}
\begin{aligned}
    \bm{E}_1^{(1)} &= \sum_{i=1}^{n_1} [(\bm{z}_i' - \bm{z}_i) \bm{z}_i'^T \ell''(\bm{z}_i^T \bm{\theta}_1;\bm{y}_{i1}) + \bm{z}_i (\bm{z}_i'^T - \bm{z}_i^T) \ell''(\bm{z}_i^T \bm{\theta}_1;\bm{y}_{i1})]
\end{aligned}
\end{equation*}
Observe that the rank of $(\bm{z}_i' - \bm{z}) \bm{z}_i'^T \ell''(\bm{z}_i^T \bm{\theta}_1;\bm{y}_{i1})$ is 1 and its single non-zero singular value is 
\begin{align*}
    \sigma &= \max_{\|\bm{x}\|_2=1} \| (\bm{z}_i' - \bm{z}) \bm{z}_i'^T \ell''(\bm{z}_i^T \bm{\theta}_1;\bm{y}_{i1}) \bm{x}\|_2 \\
    &=|\ell''(\bm{z}_i^T \bm{\theta}_1;\bm{y}_{i1})| \max_{\|\bm{x}\|_2=1} \sqrt{\bm{x}^T \bm{z}_i' (\bm{z}_i' - \bm{z}_i)^T \cdot (\bm{z}_i' - \bm{z}_i) \bm{z}_i'^T \bm{x}} \\
    &=|\ell''(\bm{z}_i^T \bm{\theta}_1;\bm{y}_{i1})| \|\bm{z}_i' - \bm{z}_i\|_2 \|\bm{z}_i'\|_2 \\
    &\leq c_2\|\bm{z}_i' - \bm{z}\|_2
\end{align*}
Similarly, $\bm{z}_i (\bm{z}_i'^T - \bm{z}_i^T) \ell''(\bm{z}_i^T \bm{\theta}_1;\bm{y}_{i1})$ has only one non-zero singular value that can be bounded by $c_2\|\bm{z}_i' - \bm{z}\|_2$. 
According to Theorem~\ref{theorem:sum singular values}, the sum of singular values of $\bm{E}_1^{(1)}$ can be bounded by
\begin{equation}\label{eq:singular values of E_1^(1)}
\begin{aligned}
    \sum_j^d \sigma_j(\bm{E}_1^{(1)}) \leq \sum_i^{n_1} 2c_2\|\bm{z}_i' - \bm{z}_i\|_2 \leq 2c_2 \Psi(\bm{Z})
\end{aligned}
\end{equation}
For $\bm{E}_1^{(2)}$, we rewrite it as $\bm{E}_1^{(2)} = \sum_{i=1}^{n_1} \bm{E}_{1, i}^{(2)}$ where
\begin{equation*}
\bm{E}_{1,i}^{(2)} \triangleq \bm{z}_i' \bm{z}_i'^T[\ell''(\bm{z}_i'^T \bm{\theta}_1;\bm{y}_{i1}) - \ell''(\bm{z}_i^T\bm{\theta}_1;\bm{y}_{i1})].
\end{equation*}
observe that $\bm{z}_i' \bm{z}_i'^T$ is of rank 1 and $\ell''$ is a scalar, so $E_{1,i}^{(2)}$ has rank $\leq 1$. $\ell''$ is Lipschitz continuous and let $c_3$ be its Lipschitz constant. When $\|\bm{\theta}_1\|_2 \leq c_{\theta}$ where $c_{\theta}$ is a constant introduced later in Eq.~(\ref{eq:c_theta}), we can bound the single non-zero singular value of $\bm{E}_{1,i}^{(2)}$ by 
\begin{align*}
    \sigma(\bm{E}_{1,i}^{(2)}) &= |\ell''(\bm{z}_i'^T \bm{\theta}_1;\bm{y}_{i1}) - \ell''(\bm{z}_i^T \bm{\theta}_1;\bm{y}_{i1})| \|\bm{z}_i'^T\|_2 \|\bm{z}_i'\|_2\\
    &\leq |\ell''(\bm{z}_i'^T \bm{\theta}_1;\bm{y}_{i1}) - \ell''(\bm{z}_i^T \bm{\theta}_1;\bm{y}_{i1})|  \\
    & \leq c_3 \|\bm{\theta}_1\|_2 \|\bm{z}_i' - \bm{z}_i\|_2 \\
    &\leq c_3 c_{\theta} \| \bm{z}_i' - \bm{z}_i \|_2. 
\end{align*}
Hence,
\begin{equation}\label{eq:singular values of E_1^(2)}
\begin{aligned}
    \sum_l^d \sigma_l(\bm{E}_1^{(2)}) 
    \leq \sum_i^{n_1} c_3 c_{\theta} \|\bm{z}_i' - \bm{z}_i \|_2 
    \leq c_3 c_{\theta} \Psi(\bm{Z}) ,
\end{aligned}    
\end{equation}
where the first inequality applies Theorem~\ref{theorem:sum singular values}. Combining Eq.~(\ref{eq:singular values of E_1^(1)}) \& (\ref{eq:singular values of E_1^(2)}), we get
\begin{equation}\label{eq:singular values of E}
\begin{aligned}
    \sum_i^d \sigma_i(\bm{E}_1) &\leq \sum_j^d \sigma_j(\bm{E}_1^{(1)}) + \sum_l^d \sigma_l(\bm{E}_1^{(2)}) \\
    &\leq (2c_2 + c_3 c_{\theta}) \Psi(\bm{Z}) ,\\
\end{aligned}
\end{equation}
where the first inequality applies Corollary 3.4.3 in \cite{horn1991topics}.

Combining Eq.~(\ref{ineq:singular values of product}), (\ref{eq:singular values of B}), \& (\ref{eq:singular values of E})
\begin{equation}\label{eq:Jacobian ratio of theta to b}
    \frac{|\det(\bm{J}(\bm{\theta}_1 \mapsto \bm{b}_1 | D)|^{-1} }{|\det(\bm{J}(\bm{\theta}_1 \mapsto \bm{b}_1' | D')|^{-1} } 
    \leq \left(1 + \frac{(2c_2 + c_3 c_{\theta}) \Psi(\bm{Z})}{d n_1(\Lambda + \Lambda')} \right)^{d}
\end{equation}
The upper bound shown above also applies to any other 
\\ \noindent$\frac{|\det(\bm{J}(\bm{\theta}_j \mapsto \bm{b}_j | D)|^{-1} }{|\det(\bm{J}(\bm{\theta}_j \mapsto \bm{b}_j' | D')|^{-1} }$, $j \in [1,c]$. By Eq.~(\ref{eq:Theta to product of theta}) \& (\ref{eq:Jacobian ratio of theta to b}), we conclude that Lemma~\ref{lemma:bound for Jacobian ratio} holds.

\end{proof}

\subsection{Lemma~\ref{lemma:bound for mu ratio}}\label{app:proof of mu ratio}
\begin{lemma}\label{lemma:bound for mu ratio}
Following the setting in Theorem~\ref{theorem:train DP}, when $\|\bm{\theta}_j\| \leq c_{\theta}$ for all $j\in [1,c]$,
\begin{equation*}
\begin{aligned}
    \frac{\mu(\bm{B} | D)}{\mu(\bm{B}' | D')}
    &\leq \exp{\left( c (c_1+c_2 c_{\theta})\Psi(\bm{Z}) \beta \right)} ,
\end{aligned}
\end{equation*}
where $c_1$ and $\beta$ follow Eq.~(\ref{eq:supremum of l(x;y)}) \& (\ref{eq:beta}), respectively.
\end{lemma}

\begin{proof}
Since $\bm{b}_1, \bm{b}_2, \cdots, \bm{b}_c$ are sampled independently and $\|\bm{\theta}_j\|_2 \leq c_{\theta}$ for all $j\in[1,c]$ in the case (i), we have
\begin{equation}\label{eq:mu B to mu b}
\begin{aligned}
    \frac{\mu(\bm{B} | D)}{\mu(\bm{B}' | D')}
    &= \prod^c_j \frac{\mu(\bm{b}_j | D)}{\mu(\bm{b}_j' | D')} .
\end{aligned}
\end{equation}
According to Eq.~(\ref{eq:bj}), for any $j \in [1,C]$,
\begin{align*}
\|\bm{b}_j' - \bm{b}_j\|_2 =& \|\sum_{i=1}^{n_1} \Big(\bm{z}_i' \ell'(\bm{z}_i'^T\bm{\theta}_j;\bm{y}_{ij}) - \bm{z}_i \ell'(\bm{z}_i^T \bm{\theta}_j;\bm{y}_{ij})\Big) \|_2\\
=& \| \sum_{i=1}^{n_1} \Big( \bm{z}_i'(\ell'(\bm{z}_i'^T \bm{\theta}_j;\bm{y}_{ij}) - \ell'(\bm{z}_i^T \bm{\theta}_j;\bm{y}_{ij})) \\
&+ (\bm{z}_i' - \bm{z}_i)\ell'(\bm{z}_i^T \bm{\theta}_j;\bm{y}_{ij})\Big) \|_2  \\
\leq& \sum_{i=1}^{n_1} [ c_2 c_{\theta} \|\bm{z}_i' - \bm{z}_i\|_2 + c_1 \|\bm{z}_i' - \bm{z}_i\|_2 ] \\
\leq& (c_1+c_2 c_{\theta})\Psi(\bm{Z}),
\end{align*}
As $\bm{b}_j$ and $\bm{b}_j'$ are sampled with probability proportional to 

\noindent $\|\bm{b}\|_2^{d-1} \exp{(-\beta \|\bm{b}\|_2)}$, we have
\begin{equation}\label{eq:bound of mu bj}
\begin{aligned}
    \frac{\mu(\bm{b}_j | D)}{\mu(\bm{b}_j' | D')} 
    & = \frac{\mu(\|\bm{b}_j\|_2 | D)/\text{surf}(\|\bm{b}_j\|_2) }{\mu(\|\bm{b}_j'\|_2 | D')/\text{surf}(\|\bm{b}_j'\|_2)} \\
    &= \frac {\|\bm{b}_j\|_2^{d-1} e^{-\beta \| \bm{b}_j\|_2}/\text{surf}(\|\bm{b}_j\|_2)} {\|\bm{b}_j'\|_2^{d-1} e^{-\beta \| \bm{b}_j'\|_2}/\text{surf}(\|\bm{b}_j'\|_2)}\\
    &= \exp{\left(\beta(\| \bm{b}_j'\|_2 - \|\bm{b}_j\|_2)\right)} \\
    &\leq \exp{\left(\beta \|\bm{b}_j' - \bm{b}_j\|_2\right)} \\
    &\leq \exp{\left(\beta (c_1+c_2 c_{\theta})\Psi(\bm{Z}) \right)}
\end{aligned}
\end{equation}
where $\text{surf}(r)$ denotes the surface area of the sphere in $d$ dimensions with radius $r$. By Eq.~(\ref{eq:mu B to mu b}) \& (\ref{eq:bound of mu bj}), we conclude that Lemma~\ref{lemma:bound for mu ratio} holds.
\end{proof}

\subsection{Lemma~\ref{lemma:bad case}}\label{app:bad case}
\begin{lemma}\label{lemma:bad case}
Following the setting in Theorem~\ref{theorem:train DP}, we have
\begin{equation}\label{ineq:bad case bound}
    \Pr[\cup_{j=1}^{c} \mathds{1}(\|\bm{\theta}_{j}\|_2 > c_{\theta})=1] \leq \sum_{j=1}^{c} \frac{\delta}{c} = \delta ,
\end{equation}
where $\mathds{1}$ is an indicator function that $\mathds{1}(\cdot)=1$ denotes the occurrence of the event $\cdot$, and $\mathds{1}(\cdot)=0$ otherwise.
\end{lemma}

\begin{proof}
By the union bound, 
\begin{equation}\label{ineq:union bound}
    \Pr[\cup_{j=1}^{c} \mathds{1}(\|\bm{\theta}_{j}\| > c_{\theta})=1] 
    \leq \sum_{j=1}^{c} \Pr[\|\bm{\theta}_{j}\| > c_{\theta}] .
\end{equation}
Since $\|\bm{z}_i\|_2 \leq 1$ and $|\ell'| \leq c_1$, when $\|\bm{\theta}_{j}\|_2 > c_{\theta}$, by Eq. (\ref{eq:bj}),
\begin{equation*}
\|\bm{b}_j\|_2 \geq n_1 (\Lambda + \Lambda') c_{\theta} - n_1 c_1 \geq n_1(\Lambda c_{\theta} - c_1) . 
\end{equation*}
Hence,
\begin{equation}\label{ineq:theta to b}
\begin{split}
    \Pr[\|\bm{\theta}_{j}\| > c_{\theta}]
    &\leq \Pr[\|\bm{b}_j\|_2 > n_1( \Lambda c_{\theta} - c_1)] \\
    &= \Pr[\|\bm{b}_j\|_2 > \frac{\beta n_1( \Lambda c_{\theta} - c_1)}{\beta}] .
\end{split}
\end{equation}
Then we substitute $\beta, c_{\theta}$ with Eq. (\ref{eq:beta}) and (\ref{eq:c_theta}),
\begin{equation*}
\begin{split}
    \beta n_1( \Lambda c_{\theta} - c_1) = \frac{\max(\epsilon - \epsilon_{\Lambda}, \omega \epsilon)}{\omega \epsilon} c_{sf} \geq c_{sf} .
\end{split}
\end{equation*}
Hence,
\begin{equation}\label{ineq:prob b to c_sf}
    \Pr[\|\bm{b}_j\|_2 > n_1( \Lambda c_{\theta} - c_1)] \leq \Pr[\|\bm{b}_j\|_2 > \frac{c_{sf}}{\beta}] .
\end{equation}
Since the distribution of $\|\bm{b}_j\|_2$ follows Eq. (\ref{eq:norm b distribution}),
\begin{equation}\label{ineq:c_sf to delta}
\begin{split}
    \Pr[\|\bm{b}_j\|_2 > \frac{c_{sf}}{\beta}] &= 1 - \int_0^{\frac{c_{sf}}{\beta}} \frac{x^{d-1} e^{-\beta x}\beta^{d}}{(d-1)!} dx \\
    &= 1 - \int_0^{c_{sf}} \frac{x^{d-1} e^{-x}}{(d-1)!} dx \leq \frac{\delta}{c} ,
\end{split}
\end{equation}
where the last inequality substitutes $c_{sf}$ with Eq. (\ref{eq:find minimal c_sf}).
Combining Eq. (\ref{ineq:union bound}), (\ref{ineq:theta to b}), (\ref{ineq:prob b to c_sf}), and (\ref{ineq:c_sf to delta}), we get
\begin{equation*}
    \Pr[\cup_{j=1}^{c} \mathds{1}(\|\bm{\theta}_{j}\|_2 > c_{\theta})=1] \leq \sum_{j=1}^{c} \frac{\delta}{c} = \delta .
\end{equation*}
\end{proof}

\subsection{\texorpdfstring{\((\epsilon,\delta)\)}{}-Probabilistic Differential Privacy}\label{app:pDP}
\begin{definition}[($\epsilon,\delta$)-Probabilistic Differential Privacy (pDP)~\cite{zhao2019reviewing}] \label{def:pDP}
a randomized algorithm $\mathcal{A}$ on domain $\mathcal{D}$ satisfies {\it ($\epsilon, \delta$)-pDP} if given
any neighboring datasets $D, D'\subseteq \mathcal{D}$, we have
\begin{equation*} \label{eq:pDP}
\mathbb{P}_{o \sim \mathcal{A}(D)}\left[e^{-\epsilon} \leq \frac{g\left(\mathcal{A}(D)=o\right)}{g\left(\mathcal{A}(D')=o\right)} \leq e^\epsilon\right] \geq 1-\delta,
\end{equation*}
where $o \sim \mathcal{A}(D)$ denotes that $o$ follows the probabilistic distribution of the output $\mathcal{A}(D)$, and $g$ denotes the probability density function.
\end{definition}

\begin{lemma}[Lemma 3 in \cite{zhao2019reviewing}]\label{lemma:pDP to DP}
    ($\epsilon, \delta$)-pDP implies ($\epsilon, \delta$)-DP.
\end{lemma}

\subsection{Singular Values of Sum of Matrices}\label{app:singular values of sum}
\begin{theorem}[Corollary 3.4.3 in \cite{horn1991topics}]\label{theorem:sum singular values}
For any $\bm{A},\bm{B} \in \mathbb{R}^{mxn}$, let $\sigma_1(A)\geq \cdots \geq \sigma_q(A)$ and $\sigma_1(B)\geq \cdots \geq \sigma_q(B)$ be their respective ordered singular values, where $q=\min(m,n)$. Let $\sigma_1(\bm{A}+\bm{B})\geq \cdots \geq \sigma_q(\bm{A}+\bm{B})$ be the ordered singular values of $\bm{A}+\bm{B}$. Then
\begin{equation*}
    \sum_{i=1}^l \sigma_i(\bm{A}+\bm{B}) \leq \sum_{i=1}^l \sigma_i(\bm{A})+\sum_{i=1}^l \sigma_i(\bm{B}), \quad l=1, \ldots, q
\end{equation*}
\end{theorem}

\subsection{Singular Values of Product of Matrices}\label{app:singular values of product}
\begin{lemma}\label{lemma:singular values of product}
    For the matrices $\bm{M}_1, \bm{M}_2$, and $\bm{M}_1 \bm{M}_2$, their singular values satisfy
    \begin{equation*}
        \sigma_i(\bm{M}_1 \bm{M}_2) \leq \sigma_1(\bm{M}_1) \sigma_i(\bm{M}_2)
    \end{equation*}
    where $\sigma_i(\cdot), \forall i \in [1,N]$ denotes the eigenvalues of matrix $\cdot$ ordered in a non-increasing order.
\end{lemma}
\begin{proof}
Using the min-max theorem of singular values~\cite{courant2008methods} of $\bm{M}_1\bm{M}_2$, we have
\begin{align*}
\sigma_i(\bm{M}_1\bm{M}_2)&=\max _{S: \operatorname{dim}(S)=i} \min_{\bm{x} \in S,\|\bm{x}\|=1} \|\bm{M}_1\bm{M}_2 \bm{x}\| \\
&\leq \max _{S: \operatorname{dim}(S)=i} \min_{\bm{x} \in S, \|\bm{x}\|=1} \|\bm{M}_1\| \cdot\|\bm{M}_2 \bm{x}\| \\
&=\sigma_1(\bm{M}_1) \cdot \max _{S: \operatorname{dim}(S)=i}  \min_{\bm{x} \in S,\|\bm{x}\|=1} \|\bm{M}_2 \bm{x}\| \\
&=\sigma_1(\bm{M}_1) \sigma_i(\bm{M}_2) .
\end{align*}
\end{proof}

In this section, we present additional experiment settings and results.
\subsection{Datasets.}\label{app:Datasets Statistics}
The primary information of the used datasets is presented in Table~\ref{tab:datasets statistics}, where Homo. ratio denotes the homophily ratio. Following the convention~\cite{kolluri2022lpgnet}, we use fixed training, validation, and testing sets for Cora-ML, CiteSeer, and PubMed. These include 20 samples per class in the training set, 500 samples in the validation set, and 1000 samples in the testing set. For Actor, we generate 5 random training/validation/testing splits with proportions of 60\%/20\%/20\%, respectively, following~\cite{huang2024optimizing}.

\begin{table}[h]
\caption{Statistics of the datasets}
\centering
\setlength\tabcolsep{1.5pt}
\begin{tabular}{l|ccccc}
\toprule
\text { Dataset } & \text { Vertices } & \text { Edges } & \text { Features } & \text { Classes } & \text { Homo. ratio } \\
\midrule
\text { Cora-ML } & 2995 & 16,316 & 2,879 & 7 & 0.81 \\
\text { CiteSeer } & 3,327 & 9,104 & 3,703 & 6 & 0.71 \\
\text { PubMed } & 19,717 & 88,648 & 500 & 3 & 0.79  \\
\text { Actor } & 7600 & 30,019 & 932 & 5 & 0.22 \\
\bottomrule
\end{tabular}
\label{tab:datasets statistics}
\end{table}

\begin{definition}[Homophily ratio] \label{def:homo ratio}
Given a graph $G$ comprising vertices $V$, edges $E$, and labels $Y$, its homophily ratio is the proportion of edges connecting vertices with the same label. Formally,
\begin{equation*}
    \text{Homo. ratio} = \frac{1}{|V|}\sum_{v\in V} \frac{1}{|\mathcal{N}_v|}\sum_{u\in \mathcal{N}_v} \mathds{1}(Y_u=Y_v),
\end{equation*}
where $\mathcal{N}_v$ is the set of neighbors of $v$, $\mathds{1}$ is the indicator function, and $Y_u$ and $Y_v$ denote the labels of vertices $u$ and $v$, respectively.
\end{definition}

\subsection{Hyperparameters}\label{app:Hyperparameters}
We employ a 2-layer fully connected neural network as our encoder (as mentioned in Section~\ref{sec:encoder module}) and tune the hidden units in $\{8, 16, 64\}$. We tune the expanded train set size $n_1$ in $\{n_0, n\}$, i.e., assign pseudo-labels for none of the unlabeled vertices or all unlabeled vertices. 

We search the restart probability $\alpha$ in $\{0.2, 0.4, 0.6, 0.8\}$ and the inference weight $\alpha_I$ in $\{\alpha\} \cup \{0.1, 0.9\}$. 
For the steps $\{m_i\}_{i=1}^s$, we set $s=1$ and search $m_1$ in $\{1, 2, 5, 10, \infty\}$ on Cora-ML, CiteSeer, and PubMed. For Actor, we set $s \in \{1,2,3\}$ and $\{m_i\}_{i=1}^s \subseteq \{0,1,2,5\}$. The privacy budget allocator $\omega$ (defined in Theorem~\ref{theorem:train DP}) is fixed at $0.9$.

We try the 2 loss functions introduced in Section~\ref{sec:GCON loss}. For the pseudo-Huber Loss (Eq. (\ref{eq:ph loss})), we tune the weight $\delta_l$ in $\{0.1, 0.2, 0.5\}$. For the noisy loss function (Eq. (\ref{eq:loss_priv})), we tune the regularization coefficient $\Lambda$ in $\{0.01, 0.2, 1, 2\}$. We optimize our model using the Adam optimizer~\cite{kingma2014adam}, with a learning rate of $\{0.001, 0.01\}$ and a weight decay of $10^{-5}$.

When assessing the competitors, we follow their tuning strategies to report their best performance. Following previous work~\cite{wu2022DPGCN, sajadmanesh2023gap, sajadmanesh2023progap}, we do not account for the privacy cost of hyperparameter tuning. For a fair comparison, we do not add noise during the inference phase for the baselines if they only use individual edges of the nodes (as explained in Section~\ref{sec:GCON inference}).

\end{document}